\documentclass[journal]{IEEEtran}

\makeatletter
\newcommand{\vast}{\bBigg@{4}}
\newcommand{\Vast}{\bBigg@{5}}
\makeatother

\usepackage{tikz}

\usepackage[utf8]{inputenc}
\usepackage{graphicx}
\usepackage{multirow,url}
\usepackage{comment}

\usepackage[ruled,vlined]{algorithm2e}
\usepackage[tight,footnotesize]{subfigure}
\usepackage[cmex10]{amsmath}
\usepackage{amssymb, empheq}
\usepackage{amsthm}
\usepackage{balance}
\usepackage{todonotes}
\usepackage{hieroglf}
\usepackage{tikz}
\usepackage{xcolor}
\definecolor{mygreen}{rgb}{0.31, 0.78, 0.47}
\usepackage[backref=page,colorlinks=true,linkcolor=black,citecolor=mygreen]{hyperref}
\renewcommand*{\backref}[1]{}

\newcommand{\cD}{\mathcal{D}}
\newcommand{\cE}{\mathcal{E}}
\newcommand{\cH}{\mathcal{H}}
\newcommand{\cI}{\mathcal{I}}

\newcommand{\N}{\mathcal{N}}

\newcommand{\bA}{\boldsymbol{A}}
\newcommand{\bB}{\boldsymbol{B}}
\newcommand{\bC}{\boldsymbol{C}}

\newcommand{\boldF}{\boldsymbol{F}}

\newcommand{\bI}{\boldsymbol{I}}

\newcommand{\bQ}{\boldsymbol{Q}}
\newcommand{\bR}{\boldsymbol{R}}

\newcommand{\bV}{\boldsymbol{V}}

\newcommand{\bX}{\boldsymbol{X}}

\newcommand{\ba}{\boldsymbol{a}}
\newcommand{\bb}{\boldsymbol{b}}

\newcommand{\boldf}{\boldsymbol{f}}

\newcommand{\bk}{\boldsymbol{k}}
\newcommand{\bl}{\boldsymbol{\ell}}

\newcommand{\br}{\boldsymbol{r}}

\newcommand{\bv}{\boldsymbol{v}}
\newcommand{\bs}{\boldsymbol{s}}
\newcommand{\bt}{\boldsymbol{t}}

\newcommand{\bx}{\boldsymbol{x}}
\newcommand{\by}{\boldsymbol{y}}
\newcommand{\bz}{\boldsymbol{z}}

\newcommand{\balpha}{\boldsymbol{\alpha}}
\newcommand{\bbeta}{\boldsymbol{\beta}}
\newcommand{\bgamma}{\boldsymbol{\gamma}}

\newcommand{\bSig}{\boldsymbol{\Sigma}}

\newcommand{\bmu}{\boldsymbol{\mu}}
\newcommand{\bzeta}{\boldsymbol{\zeta}}

\def\Zero{\boldsymbol{0}}

\newcommand{\tr}{\text{Tr}}
\newcommand{\vc}{\text{vec}}

\newcommand{\col}{\text{col}}
\newcommand{\ie}{i.e.}
\newcommand{\eg}{e.g.}

\newcommand{\expec}{\mathbb{E}}

\DeclareMathOperator*{\sym}{sym}
\DeclareMathOperator*{\modop}{mod}

\newtheorem{theorem}{Theorem}

\usepackage{color}

\definecolor{darkgreen}{rgb}{0., 0.4, 0.}
\definecolor{amber}{rgb}{1.0, 0.49, 0.0}
\definecolor{orange}{rgb}{1.0, 0.4, 0.0}

\title{Graph topology inference with derivative-reproducing property in RKHS: \\ algorithm and convergence analysis}

\author{Mircea~Moscu, Ricardo~A.~Borsoi, C\'edric~Richard,
Jos\'e-Carlos~M.~Bermudez
\thanks{This work was funded in part by ANR under grant ANR-19-CE48-0002.}
\thanks{M. Moscu is with the  Universit\'e  C\^ote  d'Azur,  Nice, France (e-mail: \mbox{mircea.moscu@oca.eu}), Lagrange Laboratory (CNRS, OCA).}

\thanks{R.A. Borsoi is with the Department of Electrical Engineering, Federal University of Santa Catarina (DEE--UFSC), Florian\'opolis, SC, Brazil, and with the Lagrange Laboratory, Universit\'e  C\^ote  d'Azur, Nice, France. e-mail: \mbox{raborsoi@gmail.com}.}
\thanks{C. Richard is with the  Universit\'e  C\^ote  d'Azur,  Nice, France (e-mail: \mbox{cedric.richard@unice.fr}), Lagrange Laboratory (CNRS, OCA).}
\thanks{J.-C.M. Bermudez is with the DEE--UFSC, Florian\'opolis, SC, Brazil. \mbox{e-mail}: \mbox{j.bermudez@ieee.org}.}
}

\begin{document}
\maketitle

\setcounter{tocdepth}{3}
\setcounter{secnumdepth}{5}

\newpage

\begin{abstract}
In many areas such as computational biology, finance or social sciences, knowledge of an underlying graph explaining the interactions between agents is of paramount importance but still challenging. Considering that these interactions may be based on nonlinear relationships adds further complexity to the topology inference problem. Among the latest methods that respond to this need is a topology inference one proposed by the authors, which estimates a possibly directed adjacency matrix in an online manner. Contrasting with previous approaches based on linear models, the considered model is able to explain nonlinear interactions between the agents in a network.
The novelty in the considered method is the use of a derivative-reproducing property to enforce network sparsity, while reproducing kernels are used to model the nonlinear interactions.
The aim of this paper is to present a thorough convergence analysis of this method. The analysis is proven to be sane both in the mean and mean square sense. In addition, stability conditions are devised to ensure the convergence of the analyzed method.
\end{abstract}

\begin{IEEEkeywords}
convergence analysis, kernel least mean squares, nonlinear topology inference, partial derivative sparsity
\end{IEEEkeywords}

\section{Introduction}
Knowledge of the network topology in applications such as gene regulation systems~\cite{FrazerSecond}, socio-economic interactions \cite{heiberger2018predicting}, or brain activity \cite{KramerEmergent} is of utmost importance. The main reason is that processing these data via, \eg, filtering \cite{sandryhaila2013discrete_filtering,nassif2017graph,chen2016multitask,coutino2018advances,koppel2017decentralized2} or spectral analysis \cite{Sandryhaila2013discrete,tremblay2016accelerated}, requires the network structure~\cite{cooperative2018djuric}.
Since most graph signal processing algorithms assume that the graph topology is known beforehand, significant progress has been made recently in the estimation of the graph topology from available data, be them links between regions of the brain, genes in a network, or sectors of a market economy. Moreover, the presence of nonlinear interactions in real-world applications led to the development of more general algorithms. One such proposition is the graph topology inference algorithm recently proposed by the authors~\cite{moscu2020online}, which is based on derivative-reproducing property to infer nonlinear interactions and promote sparsity. The current paper builds upon this preliminary work by proposing an in-depth and complex stability and performance analysis of this algorithm.
The ability of reproducing kernels to model nonlinear relationships between nodal signals is employed in the graph inference process, alongside with kernel dictionaries to mitigate the increasing number of kernel functions due to the online setting, while distributing the computational burden over the agents in the network. Many real world examples, such as social graphs \cite{speriosu2011twitter}, show considerable edge sparsity, which needs a \mbox{sparsity-inducing} framework based on derivative-reproducing kernels.
The present paper briefly recalls the considered algorithm and follows up with a complete convergence analysis, both in the mean and mean-square sense.

\subsection{Prior works}
Early advances in topology inference are put forward in \cite{dempster1972covariance}, where log-likelihood-based method of covariance estimation is introduced. Similarly, in \cite{friedman2008sparse} the graphical Lasso is employed in order to estimate the precision matrix from the available data. 
In \cite{segarra2017network}, the authors advocate that connectivity can be recovered from estimated spectral templates, while \cite{sardellitti2019graph} adapts the method for band-limited signals, \ie, signals whose Graph Fourier Transform is sparse. The authors of \cite{shafipour2017network} introduce a method designed for topology inference for the case when the measured signals are non-stationary.
An online adaptive algorithm is developed in \cite{moscu2019learning}, where the authors assume a linear model governing the interactions. Linearity of interactions and signal stationarity are assumed in \cite{shafipour2019online} to devise an ADMM algorithm.

The modeling of nonlinear interactions has proved to be a relevant research topic. Some existing methods of higher-order link inference yield satisfactory results when applied in optimizing power grids \cite{coutino2020self} or in predicting social behavior \cite{yang2020learning}.
In the presence of nonlinear phenomena, works such as \cite{Harring2012,Finch2015} focus on polynomial structural equation models, while the authors of \cite{lim2015operator} use their nonlinear counterparts. They, however, have some limitations, such as assuming prior knowledge of certain connections or the form of the nonlinear basis functions. Reproducing kernels have seen extensive use in graph topology inference problems. 
Such a solution was developed in \cite{lippert2009kernel}, where an unsupervised kernel-based method is implemented. One particularity of the algorithm is that it requires, as a parameter, the number of sought edges. It also offers the possibility of statistical significance testing when setting this parameter. 
Another work is \cite{shen2017topology} where kernels, chosen to best fit the available data, model nonlinear relationships between nodes based on measurements at successive time instants. 
The multi-kernel approach in \cite{zhang2017going} uses partial correlations to encode graph topology and $\ell_p$-norm regression to enhance the performance. 
In \cite{giannakis2018topology}, a thorough analysis of the kernel-based topology inference problem is given. This work focuses on capturing nonlinear and dynamic links. A review of state-of-the art methods in topology inference is in \cite{dong2019learning}.

None of these works, however, bases its rationale on kernel derivatives. Moreover, no analysis of this family of methods exists in the literature. For these reasons, we put forth the current paper, which first recalls the algorithm and then proposes an in-depth convergence analysis.

\subsection{Notations and definitions}
Normal font letters denote scalars, while boldface lowercase and uppercase letters stand for column vectors and matrices, respectively.
Uppercase calligraphic letters denote sets. Their cardinality is denoted by $\vert\cdot\vert $. The $(u,v)$-th entry of a matrix $\bX$ is denoted by $[\bX]_{u,v}$ or by $x_{uv}$. Notation $\balpha = \col \{\{\alpha_p\}_p\}$ denotes the column vector obtained by stacking all entries $\alpha_p$ for all available indexes $p$.
Finally, $\expec\{\cdot\}$ is the expectation operator. 

A graph $\mathcal{G}$ consists of a set $\N$ of $(N+1)$ nodes, and a set~$\cE$ of edges such that if nodes $m$ and $n$ are linked, then we have $(m,n) \in \cE$. For undirected graphs, these node pairs are unordered. At node level, we collect a \mbox{real-valued} signal and we organize it as a column vector $[y_1(i), \ldots, y_{N+1}(i)]^\top$, where $y_n(i)$ is the signal sample at node $n$ and time instant $i$. The adjacency matrix $\bA$ \cite{biggs1993algebraic, sandryhaila2014big}, is defined as an $(N+1) \times (N+1)$ matrix whose entries $a_{nm}$ are zero if $(m,n) \notin \cE$ and set to one otherwise. Throughout the paper we shall consider that there is no self-loop, \ie, $a_{nn}=0,$ for all $n$.

\subsection{Paper outline}

The paper is organized as follows. Sections \ref{sect:problem} and \ref{sect:online} recall the problem previously considered in \cite{moscu2020online} and the online algorithm, respectively. Section \ref{sect:analysis} then proceeds with the main contribution of the current paper, consisting of the analysis of the algorithm, both in mean and mean square sense. In order to validate the pertinence of the analysis, Section \ref{sect:experimental} showcases a set of experimental simulations. Finally, a set of concluding remarks is given in Section \ref{sect:conclusion}.  

\section{Local problem formulation}
\label{sect:problem}

Consider an $(N+1)$-node graph with adjacency matrix $\bA$ that models a system such as a brain network or a power grid. In this context, the electrical activity of different brain regions \cite{rubinov2010complex,Shen2016NonlinearSV}, or the voltage angle per bus \cite{zhang2017going}, numbered from $1$ to $(N+1)$, can be measured sequentially at different time instants $i\in\mathbb{N}$, leading to a graph signal given by $ [y_1(i), \ldots, y_{N+1}(i)]^\top\in\mathbb{R}^{N+1}$ that changes over time. Signal $y_n(i)$ at each node $n=1,\ldots,N+1$ of the graph is nonlinearly coupled to the signals at all nodes in its neighborhood, according to the topology described in $\bA$. Moreover, we assume that nonlinear relationships between signals at different nodes occur as have been reported in many applications; see, e.g., the case of brain connectivity~\cite{freeman1979eeg, de2009hemodynamic}.

Recent methods have considered additive nonlinear models~\cite{shen2017kernel,moscu2019online}, which represent the local measurements at each node $n$ as:
\begin{equation}
y_n(i) = \sum_{m \in \N \setminus \{n\} } f_{nm}(\by_{m}(i)) + \varrho_n(i)\,,
\label{eq:model_central_old}
\end{equation}
where $\by_{m}(i)=\big[y_m(i),\ldots,y_m(i-L_m+1)\big]^\top$, for $L_m\geq1$. 
Parameter $L_m$ endows the algorithm with memory-like capacities. This characteristic is desirable in a number of applications such as brain topology inference where there is a 10--20~ms delay in signal propagation between nodes~\cite{petkoski2019transmission}.
Functions $f_{nm}:\mathbb{R}^{L_m}\to\mathbb{R}$ represent the interactions between the different nodes in the network, and $\varrho_n(i)$ denotes innovation noise. By relating how each node $m\in\N \setminus \{n\}$ affects node $n$, functions $f_{nm}$ encode the connectivity in $\bA$ directly as $a_{nm}=0$ if, and only if, $f_{nm}\equiv0$, where $a_{nm}$ is the $(n,m)$-th entry of $\bA$.
For ease of notation, we shall assume that $n\equiv(N+1)$, i.e., we identify $n$ as the $(N+1)$-th node of the graph, which allows us to denote $\N\setminus \{n\}=\{1,\ldots,N\}$.
%
%
The topology inference problem then consists of finding the set of functions $f_{nm}$ which best represent the available graph signal measurements $y_n(i)$ based on the previous model. This can be performed using both batch-based (see, e.g.,~\cite{shen2017kernel}) and online strategies~\cite{moscu2019online,moscu2021convergence,shen2018online}.


Models such as \eqref{eq:model_central_old} do not consider general nonlinear interactions between multiple nodes, as they assume an additive model for $y_n(i)$~\cite{buja1989linearSmootherAdditiveModels}. This can be rather limiting since the nonlinearity introduced in the model through the functions $f_{nm}$ is {local}, i.e., it only acts on the signals at each node individually.
To overcome this issue, we propose to consider the following general nonlinear model:
\begin{equation}
    y_n(i) = f_n(\by(i)) + \varrho_n(i)\,,
    \label{eq:model_central}
\end{equation}
where $\by(i)=[\by_{1}(i)^\top,\ldots,\by_{N}(i)^\top]^\top$ and function $f_n$ describes the interaction between the signals at nodes $\N\setminus\{n\}$ and the signal $y_n(i)$ at node $n$. Compared to \eqref{eq:model_central_old}, model \eqref{eq:model_central} captures more complex relationships between nodes without relying on generalized linear model, rendering it more general.

The method we propose relates $f_n$ to the underlying graph topology $\bA$ by quantifying how a node $m$ affects node $n$ through the corresponding partial derivative, \ie:
\begin{align}
    \parbox{25ex}{node $m$ does not affect\\\centering node $n$ (\ie, $a_{mn}=0$)}
    \, \Longleftrightarrow \, 
    \bigg\| \frac{\partial f_n(\by)}{\partial \by_{m}} \bigg\|
    = 0 \,,
    \label{eq:id_condition}
\end{align}
under the assumption that $f_n$ is continuously differentiable, where the expectation is taken w.r.t.~$\by$.

Estimating the graph topology can then be performed by learning function $f_n$ that fits the data as per model~\eqref{eq:model_central}, while additionally considering penalties to promote a sparsely connected graph, as in many real-world applications~\cite{danisch2018listing}. Afterwards, the topology can be recovered via relation~\eqref{eq:id_condition}. 
By constraining $f_n$ to belong to a Reproducing Kernel Hilbert Space~(RKHS) $\cH$ associated with a positive definite reproducing kernel $\kappa(\cdot,\cdot)$, this problem can be formulated in batch form as~\cite{moscu2020online}:
\begin{align}
    & \min_{f_n \in\mathcal{H}}\,\,\,
    \dfrac{1}{2i} \sum_{\ell=1}^i 
    \Big\|y_n(\ell)- f_n(\by(\ell)) \Big\|^2 
    \nonumber \\
    & + \eta \left( \sum_{m=1}^N \sqrt{\dfrac{1}{i} \sum_{p=1}^{i} \sum_{q=1}^{L_m} \bigg(\dfrac{\partial f_{n}(\by(p))}{\partial y_{{m,q}} } \bigg)^2 } + \psi\big(\|f_n\|_{\cH}\big) \right)\,, \label{eq:sparse_model}
\end{align}
where $y_{{m,q}}$ represents the $q$-th entry of $\by_{m}$.
The first term in \eqref{eq:sparse_model} is a data fitting term that averages the reconstruction error over all available samples. The second one is a convex group-sparsity inducing penalty~\cite{bach2011optimizationSparsity,rosasco2013nonparametric}, which promotes sparsity in the estimated derivatives averaged over all observations $\by(p)$ up to time instant $i$, and consequently on adjacency matrix $\bA$. The last term, with $\psi:\mathbb{R}\to[0,\infty)$, is a monotonically increasing function of $\|f_n\|_{\cH}$ with small magnitude in order to guarantee the uniqueness of the solution of~\eqref{eq:sparse_model}. Parameter $\eta$ makes a trade off between data the fitting term and the regularizers.

Problem~\eqref{eq:sparse_model} allows us to introduce sparsity in $f_n$ without the inherent limitations of an additive model. However, the sparsity-promoting penalty term prevents the application of previously established Representation Theorems that guarantee the existence of a finite-dimensional representation of the solution.
Assuming that kernel $\kappa(\cdot,\cdot)$ is twice differentiable, the following holds~\cite{zhou2008derivativeKernels}:
\begin{align}
    \cH \ni \dfrac{\partial f_{n}(\by)}{\partial y_{m,q} } 
    = \langle f_{n}, \kappa_{\partial_{m,q}}(\cdot,\by) \rangle_{\mathcal{H}}\,,
    \label{eq:der_rkhs}
\end{align}
where:
\begin{align}
    \kappa_{\partial_{m,q}}(\cdot,\by(q)) = \frac{\partial\kappa(\cdot,\ba)}{\partial a_{m,q}} \bigg|_{\ba=\by(q)} \,.
\end{align}
Thus, for sufficiently smooth kernels, the derivative of functions in $\cH$ belongs to the same $\cH$, in turn meaning that they can be evaluated in the form of simple inner products. 
The solution of~\eqref{eq:sparse_model} can now be written in a finite dimensional form, similarly to the proposition in~\cite{rosasco2013nonparametric}. By generalizing the Representer Theorem in the aforementioned paper, we obtain a Representer Theorem for our approach. This result is stated in the following theorem:

\begin{theorem}
Let $\cH$ be a RKHS whose associated reproducing kernel $\kappa(\cdot,\cdot)$ is at least twice differentiable. Then, the solution of the optimization problem~\eqref{eq:sparse_model} can be written as:
\begin{align}
    f_n^{*} {}={} & \sum_{p=1}^i \alpha_{p} \kappa(\cdot,\by(p)) \nonumber \\
    & + \sum_{m=1}^N \sum_{q=1}^{i} \sum_{\ell=1}^{L_m} \beta_{m,q,\ell} \kappa_{\partial_{m,\ell}}\left(\cdot,\by(q)\right)\,.
    \label{eq:rep_theorem}
\end{align}
\label{theorem1}
\end{theorem}

\begin{proof}
The proof of this theorem is presented in Appendix~\ref{proof_theo1}.
\end{proof}

Expression~\eqref{eq:rep_theorem} can then be substituted in problem \eqref{eq:sparse_model} in order to obtain a finite-dimensional optimization problem.

In the sequel, for the sake of simplicity in the notation and ease of comprehension, we shall consider the case $L_m = 1$ for all $m$, \ie, $\by_{m}(i)=y_m(i)$ and $\kappa_{\partial_{m,q}}(\cdot,\by)=\kappa_{\partial_{m}}(\cdot,\by)$.

\section{An online algorithm}
\label{sect:online}

The number of coefficients $\alpha_p$ and $\beta_{m,q}$ in solution \eqref{eq:rep_theorem} can become prohibitive as $i$ increases with each new measurement and associated kernel function. Kernel dictionaries are a solution to this issue. These dictionaries can be defined either \textit{a priori} \cite{chen2014convergence} or can admit a new candidate kernel function only if it satisfies a certain sparsification rule \cite{gao2013kernel}, such as the approximate linear dependence criterion \cite{engel2004the,bueno2020gram}, the Nyström method \cite{williams2001using}, dictionary-based random Fourier features \cite{bouboulis2016efficient}, orthogonal projection \cite{takizawa2015adaptive} or the novelty criterion~\cite{platt1991resource}.

We consider a dictionary-based framework, where each node $n$ in the network is equipped with a dictionary of kernel functions (and their derivatives) of the following form:
\begin{align}
\cD_n = \big\{ \kappa (\cdot,\by(\omega_j)) : \omega_j \in \cI^{i}_n \big\} 
\,,
\label{eq:dict_def}
\end{align}
where $\cI^{i}_n\subset \{1,\ldots,i-1\}$ represents the set of time indices of elements selected for the dictionary, before instant $i$. To build $\cD_n$ iteratively, we consider the coherence criterion as the chosen method for dictionary sparsification~\cite{richard2009online}.
This entails the fact that, after a sufficient number of samples $i$, only a small number $\vert \cI^{i}_n \vert \ll i$ of coefficients will be needed. A candidate kernel function $\kappa(\cdot,\by(i))$ is added to $\cD_n$ if the following sparsification condition holds \cite{richard2009online}:
\begin{equation}
    \max_{\omega_j \in \cI^{i}_n} \vert \kappa(\by(i),\by(\omega_j))\vert \leq \xi_n \,,
    \label{eq:condition_dictio}
\end{equation}
where $\xi_n \in [0,1)$ determines the level of sparsity and the coherence of the dictionary. The number of entries in the dictionary satisfies $\vert \cI^{i}_n \vert < \infty$ when $i \to \infty$ \cite{richard2009online}.
Note that the dictionary $\cD_n$ can also be set \textit{a priori} using, e.g., an uniform grid sampling~\cite{chen2014convergence}. In this case, with a slight abuse of notation, we represent the dictionary using a set of time indices $\cI^{i}_n\subset \mathbb{Z}\setminus\mathbb{N}$, attributing to the dictionary elements negative time indexes.
We rewrite \eqref{eq:rep_theorem} as:
\begin{equation}
    \begin{split}
    f_{n,\cD_n} {}={} & \sum_{p=1}^{\vert \cD_n \vert}\alpha_{p} \kappa(\cdot,\by(\omega_p)) \\
    &+ \sum_{m=1}^{N} \sum_{q=1}^{\vert \cD_n \vert} \beta_{m,q} \kappa_{\partial_m}(\cdot,\by(\omega_q)) \,.
    \label{eq:rep_theorem_dictio}
    \end{split}
\end{equation}
For notation compactness, we shall write $\balpha = \col \{\{\alpha_p\}_p\}$, and $\bbeta = \col \{\{\bbeta_m\}_m\}$ with $\bbeta_m = \col \{\{\beta_{m,q}\}_q\}$.

The dictionary-based representation $f_{n,\cD_n}$ in~\eqref{eq:rep_theorem_dictio} can now be plugged directly into the batch problem~\eqref{eq:sparse_model} to obtain a solution depending on a small amount of coefficients, for all~$i$. However, this would still entail a growth in complexity with time $i$, since the kernel functions have to be evaluated at all available data points. To address this issue and devise a computationally scalable algorithm, we consider an alternative cost function to~\eqref{eq:sparse_model}. Specifically, we substitute the averaging operations in the data fitting term and in the derivative penalty term in~\eqref{eq:sparse_model}, which use all data points up to time instant $i$, by the corresponding expected values. Setting $\psi\equiv0$ without loss of generality, this leads to:
%
%
\begin{align}
    \min_{f_{n,\cD_n}} \,\,
    & \dfrac{1}{2}\,\expec\left\{
    \Big\|y_n(i) - f_{n,\cD_n}(\by(i)) \Big\|^2 \right\} \nonumber \\
    & + \eta \sum_{m=1}^N \sqrt{
    \expec\left\{\bigg(\frac{\partial f_{n,\cD_n}(\by(i))}{\partial y_{m}(i)}\bigg)^2 \right\} }\,.
    \label{eq:opt_problem_online2_func}
\end{align}
where the expectations are w.r.t. $\by(i)$ and conditioned on the dictionary~\eqref{eq:dict_def}.

Combining~\eqref{eq:opt_problem_online2_func} with~\eqref{eq:rep_theorem_dictio} as shown in Appendix \ref{annex:fnd}, we get:
\begin{equation}
    \begin{split}
    \min_{\bgamma} \,\,
    J(\bgamma)= & \dfrac{1}{2}\,\expec\Big\{ 
    \Big\|y_n(i) - \bgamma^\top \bs(i) \Big\|^2 \Big\} \\
    &+ \eta \sum_{m=1}^N \sqrt{\bgamma^\top \bR_{tt,m}(i)\,\bgamma}\,,
    \label{eq:opt_problem_online2}
    \end{split}
\end{equation}
with $\bs(i) = \begin{bmatrix} \bz(i) \\ \bk(i) \end{bmatrix}$, $\bgamma=\begin{bmatrix}\bbeta \\ \balpha \end{bmatrix}$, and $\bt_m(p) = \begin{bmatrix} \bl_m(p) \\ \bzeta_m(p) \end{bmatrix}$,\\ where:
\begin{align}
    &\bR_{tt,m}(i) \triangleq \expec\{\bt_m(i) \bt^\top_m(i)\} \label{eq:Rttm}\\
    &\bk(i) = \col \Big\{ \kappa (\by(i),\by(\omega_q)) \Big\}_{q=1}^{\vert \cD_n \vert} \label{eq:k} \\
    &\bz(i) = \big[\bz_1^\top\!(i),\ldots, \bz_{N}^\top\!(i)\big]^\top \\
    &[\bz_m(i)]_{q} = \dfrac{\partial \kappa(\by(i),\by(\omega_q))}{\partial y_{m}(\omega_q)}\Bigg\vert_{q=1,\ldots,\vert \cD_n \vert}\label{eq:z} \\
    &\bzeta(i) = \big[\bzeta_1^\top\!(i),\ldots,\bzeta_N^\top\!(i)\big]^\top \\
    &[\bzeta_m(i)]_{q} = \dfrac{\partial \kappa(\by(i),\by(\omega_q))}{\partial y_{m}(i)}\Bigg\vert_{q=1,\ldots,\vert \cD_n \vert}\label{eq:zeta} \\
    &\bl_m(i) = \big[\bl_{1,m}^\top\!(i), \ldots, \bl_{N,m}^\top\!(i)\big]^\top \\
    &[\bl_{m_1,m_2}(i)]_q = \dfrac{\partial^2 \kappa (\by(i),\by(\omega_q))}{\partial y_{m_1}(\omega_q) \partial y_{m_2}(i)}\Bigg\vert_{q=1,\ldots,\vert \cD_n \vert} 
    \label{eq:l}
\end{align}
Note that $\bgamma$ is now the optimization variable in problem \eqref{eq:opt_problem_online2}. Quantities \eqref{eq:k}--\eqref{eq:l} can be computed in closed form when an explicit expression for kernel~$\kappa(\cdot,\cdot)$ is provided. Expressions in the case of the Gaussian kernel are given in Appendix~\ref{annex:gaussian}. A closed form expression of $\bR_{tt,m}$ for Gaussian random variables is given in the supplementary material.

Calculating the gradient of \eqref{eq:opt_problem_online2} leads to:
\begin{align}
    \nabla J(\bgamma) = - \expec\Big\{\bs(i) &\big[y_n(i) - \bs^\top\!(i) {\bgamma}\big]\Big\} \nonumber \\
    &+ \eta \sum_{m=1}^N \dfrac{\bR_{tt,m}(i) {\bgamma}}
    {\sqrt{\bgamma^\top\,\bR_{tt,m}(i)\,\bgamma}}\,.
    \label{eq:gradient}
\end{align}
Approximating the first expectation in~\eqref{eq:gradient} by the instantaneous estimate $\bs(i)\big[y_n(i) - \bs^\top\!(i) {\bgamma}(i)\big]$, and $\bR_{tt,m}(i)$ by a sample covariance matrix $\hat{\bR}_{tt,m}(i)$ available at time $i$, allows us to devise the following subgradient descent algorithm to iteratively minimize $J(\bgamma)$:
\begin{align}
    \hat{\bgamma}(i+1) {}={} \,  \hat{\bgamma}(i)& + \mu \bs(i) \big[y_n(i) - \bs^\top\!(i) \hat{\bgamma}(i)\big] \nonumber \\
    & - \mu \eta \sum_{m=1}^N \dfrac{\hat\bR_{tt,m}(i) {\hat\bgamma(i)}}
    {\sqrt{\hat\bgamma^\top\!(i)\,\hat\bR_{tt,m}(i)\,\hat\bgamma(i)}}\,,
    \label{eq:update}
\end{align}
with $\mu$ a small positive step size. We set $x/0\triangleq0$ to write the subgradient when the denominator in the last term of \eqref{eq:update} equals $0$. For the sake of brevity, we introduce the notation:
\begin{equation}
    \label{eq:Delta}
    \Delta_m(i) 
    = \sqrt{\hat\bgamma^\top\!(i)\,\hat\bR_{tt,m}(i)\,\hat\bgamma(i)}\,.
\end{equation}
The covariance matrix $\hat\bR_{tt,m}(i)$ needs to be updated sequentially in environments where the data arrives sequentially. An extensive body of literature concerning sequential estimators exists. The following sample covariance matrix provides such an estimate:
\begin{equation}
    \label{eq:cumulative.average.0}
    \hat\bR_{tt,m}(i) = \alpha\,\hat\bR_{tt,m}(i-1) + (1-\alpha)\,\bt_m(i) \bt^\top_m(i)\,,
\end{equation}
with $\alpha\in[0,1)$ a forgetting factor. In the absence of hypothesis about the covariance matrix, a successful approach so far has been shrinkage estimation~\cite{Ledoit2012}. We shall not discuss these details here, which are out of the scope of this paper, but refer the reader to the literature. We shall however assume that the estimator is unbiased, that is:
\begin{equation}
    \label{eq:cumulative.average}
    \expec\{\hat\bR_{tt,m}(i)\} = \bR_{tt,m}(i)\,.
\end{equation}
The time index $i$ will be omitted for stationary variables.

In~\eqref{eq:update}--\eqref{eq:Delta}, observe that each $\Delta_m(i)$ represents an estimate of the squared norm of the partial derivative of $f_n$ with respect to $\by_m$. It thus allows us to infer the graph topology by comparing every quantity $\Delta_m(i)$ to a given threshold $\tau_n$ such that, following~\eqref{eq:id_condition}, $\hat{a}_{nm}(i)=1$ if $\Delta_m(i)\geq\tau_n$ and $0$ otherwise. These thresholds can be chosen to either obtain an estimated topology which realistically explains the studied process \cite{Shen2016NonlinearSV}, or to obtain a connected graph, \ie, a graph in which there exists a path between all pairs of nodes.


\section{Theoretical analysis}
\label{sect:analysis}

Let $\bv(i) = \hat{\bgamma}(i) - \bgamma^*$ be the weight-error vector, where $\hat{\bgamma}(i)$ is the current estimate in \eqref{eq:update} and $\bgamma^*$ is the
%
optimal coefficient vector, obtained as the solution to problem~\eqref{eq:opt_problem_online2}.
%

\subsection{Relating the error in derivatives to the error in coefficients}

One difficulty related to the iterative algorithm~\eqref{eq:update} is that it considers the evolution of the coefficients $\bgamma$, instead of the partial derivatives ${\partial \hat{\!f}_{n,\mathcal{D}_n}(\by)}/{\partial y_m}$. Nevertheless, we can study the convergence of the derivatives by means of filter coefficients as will be shown in the following.

Using \eqref{eq:der_rkhs} and the Cauchy-Schwarz inequality:
\begin{align}
    &\left|\dfrac{\partial \hat{\!f}_{n,\mathcal{D}_n}(\by)}{\partial y_m} - \dfrac{\partial f_{n,\mathcal{D}_n}^*(\by)}{\partial y_m}\right| = \\ &= \left|
    \langle \hat{\!f}_{n,\mathcal{D}_n}, \kappa_{\partial_{m}}(\cdot,\by) \rangle_{\mathcal{H}} - \langle f^*_{n,\mathcal{D}_n}, \kappa_{\partial_{m}}(\cdot,\by) \rangle_{\mathcal{H}} \right|
    \nonumber \\
    &= \left| \langle \hat{\!f}_{n,\mathcal{D}_n}-f^*_{n,\mathcal{D}_n}, \kappa_{\partial_{m}}(\cdot,\by) \rangle_{\mathcal{H}} \right|
    \nonumber \\
    & \leq \Vert \kappa_{\partial_{m}}(\cdot,\by) \Vert_{\mathcal{H}} \Vert \hat{\!f}_{n,\mathcal{D}_n} - f^*_{n,\mathcal{D}_n} \Vert_{\mathcal{H}} \,.
\end{align}
Moreover, we also have:
\begin{align}
    \Vert \hat{\!f}_{n,\mathcal{D}_n} - f^*_{n,\mathcal{D}_n} \Vert_{\mathcal{H}} &= \Vert \bs^\top \hat{\bgamma} - \bs^\top \bgamma^* \Vert
   =  \Vert \bs^\top (\hat{\bgamma} - \bgamma^*) \Vert
    \nonumber \\
    & \leq \Vert \bs \Vert \,\Vert \hat{\bgamma} - \bgamma^* \Vert\,.
\end{align}
By taking the expectation of both sides conditioned on $\hat{\bgamma}$
, we have:
\begin{align}
    \expec\Bigg\{\bigg| & \dfrac{\partial \hat{\!f}_{n,\mathcal{D}_n}(\by)}{\partial y_m} -  \dfrac{\partial f_{n,\mathcal{D}_n}^*(\by)}{\partial y_m}\bigg|
    \,\Bigg|\, \hat{\bgamma} 
    \Bigg\} \nonumber \\ & \leq \expec\bigg\{\Vert \kappa_{\partial_{m}}(\cdot,\by) \Vert_{\mathcal{H}} \Vert \bs \Vert  \,\bigg|\, \hat{\bgamma} \bigg\} \, 
    \Vert \hat{\bgamma} - \bgamma^* \Vert
    \,.
\end{align}
where the first expectation on the r.h.s. is not related to $\hat{\!f}_{n,\mathcal{D}_n}$ and $f^*_{n,\mathcal{D}_n}$ but depends on the statistics of the input data and on the reproducing kernel and its derivatives. Thus the error in the estimated derivatives is bounded by the error in the coefficients~$\hat\bgamma$. This allows us to study the convergence behavior of the derivatives indirectly by means of the coefficients $\hat{\bgamma}$.

\subsection{Simplifying assumptions}

To make the analysis of algorithm~\eqref{eq:update} tractable, we need to introduce some simplifying assumptions.

\begin{enumerate}
    \item[\textsf{A1)}] We assume that vector $\by(i)$ is stationary  and Gaussian-distributed, with zero-mean and covariance matrix $\bR_{y}$, namely:
    \begin{equation*}
        \hspace{-0.7cm} 
        p\big(\by(i)\big) = (2\pi)^{-N/2}\det\{\bR_{y}\}^{-1/2}\exp\Big(-\frac{1}{2} \by^\top\!(i) \bR_{y}^{-1} \by(i)\Big)\,.
    \end{equation*}
    
    \item[\textsf{A2)}] We assume that $L_m=1$ for simplicity, but the same analysis can be extended for $L_m>1$.
    
    \item[\textsf{A3)}] We consider the Gaussian kernel: 
    \begin{equation}
    \kappa(\ba,\bb) = \exp \left( - {\Vert \ba - \bb \Vert^2}/{2\sigma^2} \right)\,,
    \label{eq:gaussian_ker}
    \end{equation}
    due to its capacities as an universal approximator \cite{liu2010kernel}.
    
    \item[\textsf{A4)}] We make the Modified Independence Assumption, that is, $\bs(i)\bs^\top\!(i)$ is statistically independent of $\bv(i)$. This assumption has been successfully used in the analysis of various adaptive filtering algorithms~\cite{parreira2012stochasticGaussianKLMS}. It has also been shown in~\cite{minkoff2001comment} to be less restrictive when compared to the classical independence assumption.
    
    \item[\textsf{A5)}] The elements $\{\by(\omega_p)\}$ in the dictionary are set \textit{a priori}. They are thus independent of $\by(i)$.
    
    \item[\textsf{A6)}] The entries of $\bv(i)$ are jointly Gaussian distributed. This hypothesis has been used and validated in the analysis of both linear~\cite{chen2009sparse,chen2016transientZeroAttLMS} and nonlinear (kernel-based)~\cite{gao2019convergenceSparseKLMS} adaptive algorithms.
\end{enumerate}
We  also use the following approximation
\begin{equation}
    \hat\bR_{tt,m}(i)\approx \expec\big\{ \hat\bR_{tt,m}(i) \big\}\stackrel{\eqref{eq:cumulative.average}}{=} \bR_{tt,m}
\end{equation}
to make the analysis tractable. As an example, note that the decay of the estimation error of $\bR_{tt,m}$ using \eqref{eq:cumulative.average.0} with i.i.d. samples is in the order of $(1/i)$.

\subsection{Mean weight error analysis}
The estimation error $e(i)$ is given by:
\begin{align}
    e(i) & \triangleq y_n(i) - f_n(\by(i)) 
    = y_n(i) - \bs^\top\!(i) \hat{\bgamma}(i)
    \nonumber \\
    & = \underbrace{y_n(i) - \bs^\top\!(i) \bgamma^*}_{e_{0}(i)} - \,\bs^\top\!(i) \bv(i)\,.
    \label{eq:error} 
\end{align}
Considering \eqref{eq:update} and \eqref{eq:error}, we can write $\bv(i+1)$ as:
\begin{align}
    \bv (i+1)& 
     =  \bv(i)  - \mu \bs(i)\bs^\top\!(i) \bv(i)  + \mu \bs(i) e_0(i) \label{eq:weight_error} \\
    & - \mu \eta \sum_{m=1}^N \dfrac{\hat\bR_{tt,m}(i)\left[\bv(i) + \bgamma^*\right]}{\sqrt{\left[\bv(i) + \bgamma^*\right]^\top \hat\bR_{tt,m}(i) \left[\bv(i) + \bgamma^*\right]}}\,,
    \nonumber
\end{align}
where $e_{0}(i)=y_n(i)-\bs^\top\!(i)\bgamma^*$ is the optimal estimation error. 
Taking the expected value of both sides of \eqref{eq:weight_error}, and using \textsf{A4}, we obtain the mean weight error model:
\begin{align}
    \expec & \{\bv(i+1)\} = (\bI - \mu \bR_{ss}) \expec \{\bv(i)\} 
    + \mu\big(\br_{sy} - \bR_{ss}\bgamma^*\big) \nonumber \\
    & - \mu \eta \sum_{m=1}^N \expec\Bigg\{\dfrac{\hat\bR_{tt,m}(i) \left[\bv(i) + \bgamma^*\right]}{\sqrt{\left[\bv(i) + \bgamma^*\right]^\top \hat\bR_{tt,m}(i) \left[\bv(i) + \bgamma^*\right]}}\Bigg\}\,,
    \label{eq:expec_weight_error}
\end{align}
with $\bR_{ss} \triangleq \expec \{ \bs(i) \bs^\top\!(i) \}$ and $\br_{sy} \triangleq \expec\{\bs(i)y_n(i)\}$.

The last term in \eqref{eq:expec_weight_error} is a non-trivial expectation involving $\bv(i)$ and $\hat\bR_{tt,m}(i)$. To proceed, we need to use the simplifying assumption $\hat\bR_{tt,m}(i)\approx \bR_{tt,m}$ discussed above. This yields:
\begin{align}
    \expec\left\{\dfrac{\hat\bR_{tt,m}(i) \hat{\bgamma}(i)}{\sqrt{\hat{\bgamma}^\top\!(i) \hat\bR_{tt,m}(i)\hat{\bgamma}(i)}}\right\}
    \approx
    \expec\left\{\dfrac{\bR_{tt,m} \hat{\bgamma}(i)}{\sqrt{\hat{\bgamma}^\top\!(i) \bR_{tt,m}\hat{\bgamma}(i)}}
    \right\}\,.
\end{align}
By successively approximating the expectation of the ratio by the ratio of the expectations, and the expectation of the square root by the square root of the expectation \cite[p. 70]{elandt1999survival}, we have:
\begin{align}
    \expec&\left\{\dfrac{\bR_{tt,m} \hat{\bgamma}(i)}{\sqrt{\hat{\bgamma}^\top\!(i) \bR_{tt,m}\hat{\bgamma}(i)}}
    \right\} 
    \nonumber \\
    &\approx \dfrac{\bR_{tt,m}\bmu(i)}{\sqrt{\tr\{\bR_{tt,m}\bSig(i)\}+\bmu^\top\!(i)\bR_{tt,m}\bmu(i)}}\,,
    \label{eq:approx_reg}
\end{align}
which represents the sought-after approximation of the regularization term, with $\bmu(i)$ and $\bSig(i)$ given by:
\begin{equation}
    \label{eq:mu}
    \bmu(i) = \expec\big\{\hat{\bgamma}(i)\big\}
    = \expec\{\bv(i)\} + \bgamma^*\,,
\end{equation}
and:
\begin{align}
    \bSig(i) &= \expec\big\{\hat{\bgamma}(i)\hat{\bgamma}^\top\!(i)\big\}-\bmu(i)\,\bmu^\top\!(i)
    \nonumber \\
    &= \expec\{\bv(i)\bv^\top\!(i)\} - \expec\{\bv(i)\}\expec\{\bv(i)\}^\top\,,
    \label{eq:Sigma}
\end{align}
respectively. Using these results in~\eqref{eq:expec_weight_error}, we obtain the following expression for the mean error recursion:
\begin{align}
    \expec& \{\bv(i+1)\} \approx (\bI - \mu \bR_{ss}) \expec \{\bv(i)\}
    + \mu\big(\br_{sy}-\bR_{ss}\bgamma^*\big) \nonumber \\
    & - \mu \eta \sum_{m=1}^N \dfrac{\bR_{tt,m}\,\bmu(i)}{\sqrt{\tr\{\bR_{tt,m}\bSig(i)\}+\bmu^\top\!(i)\,\bR_{tt,m}\,\bmu(i)}}\,.
    \label{eq:expec_weight_error2}
\end{align}
where $\bmu(i)$ and $\bSig(i)$ have to be calculated at each time~$i$ in order to make recursion~\eqref{eq:expec_weight_error2} functional. Expressions~\eqref{eq:mu} and~\eqref{eq:Sigma} are used, respectively. Specifically, a recursive model is provided in subsection~\ref{ss:mean_squared} for calculating $\expec\{\bv(i)\bv^\top\!(i)\}$.



\subsection{Mean stability analysis}

Iterating \eqref{eq:expec_weight_error} starting from $i=0$, we arrive to the following expression:
\begin{align}
    \expec & \{\bv(i+1)\} \nonumber \\ 
    =\, & (\bI - \mu \bR_{ss})^{i+1} \expec \{\bv(0)\} 
    - \mu \eta \sum_{\ell=0}^i (\bI - \mu \bR_{ss})^{i-\ell}
    \nonumber \\
    & \times\sum_{m=1}^N\expec\left\{\dfrac{\hat\bR_{tt,m}(\ell) \left[\bv(\ell) + \bgamma^*\right]}{\sqrt{\left[\bv(\ell) + \bgamma^*\right]^\top \hat\bR_{tt,m}(\ell) \left[\bv(\ell) + \bgamma^*\right]}}\right\}
    \nonumber \\
    & + \mu \sum_{\ell=0}^i (\bI - \mu \bR_{ss})^{i-\ell} (\br_{sy}-\bR_{ss}\bgamma^*)\,,
    \label{eq:weight_error_closedForm}
\end{align}
where $\expec\{\bv(0)\}$ is the initial condition. $\expec\{\bv(i+1)\}$ converges when $i\rightarrow\infty$ if, and only if, all terms on the r.h.s. of~\eqref{eq:weight_error_closedForm} converge to finite values. The first term converges  to zero as $i\rightarrow \infty$ if the matrix $(\bI - \mu \bR_{ss})$ is stable. A sufficient condition to ensure the stability of $(\bI - \mu \bR_{ss})$ is to choose the step-size $\mu$ according to:
\begin{equation}
    0 < \mu < \frac{2}{\lambda_{\max}(\bR_{ss})}\,.
    \label{eq:stab_cond}
\end{equation}
where $\lambda_{\max}(\cdot)$ denotes the maximum eigenvalue of its matrix argument. We shall now prove the convergence of the second and third series on the r.h.s. of~\eqref{eq:weight_error_closedForm}. For compactness, we write the second series as $\sum_{\ell=0}^i (\bI-\mu\bR_{ss})^{i-\ell}\sum_{m=1}^N\boldsymbol{q}_m(\ell)$.
To prove the convergence of this series, it is sufficient to prove that the series $\sum_{\ell=0}^i\sum_{m=1}^N[(\bI-\mu\bR_{ss})^{i-\ell}\boldsymbol{q}_m(\ell)]_k$ converges for all $k$, where $[\cdot]_k$ is the $k$-th entry of its vector argument. A series is absolutely convergent if each term of the series can be bounded by a term of an absolutely convergent series. Using Jensen's inequality, we have:
\begin{align}
    & \left|[(\bI-\mu\bR_{ss})^{i-\ell}
    \boldsymbol{q}_m(\ell)]_k\right| \nonumber\\
    &\leq
    \left\|(\bI-\mu\bR_{ss})^{i-\ell}
    \boldsymbol{q}_m(\ell)\right\| \nonumber\\
    &\approx \left\| (\bI - \mu \bR_{ss})^{i-\ell}\,
     \expec\left\{\dfrac{\bR_{tt,m} \left[\bv(\ell) + \bgamma^*\right]}{\sqrt{\left[\bv(\ell) + \bgamma^*\right]^\top \bR_{tt,m} \left[\bv(\ell) + \bgamma^*\right]}}\right\}\right\|
    \nonumber \\
    &\leq \left\|\bI - \mu \bR_{ss}\right\|^{i-\ell} 
    \nonumber 
    \, \expec\left\{\dfrac{\|\bR_{tt,m}\|\,\|\bv(\ell) + \bgamma^*\|}{\|\bv(\ell) + \bgamma^*\|_{{\bR_{tt,m}}}} \right\}
    \nonumber \\
    &\leq \left\|\bI - \mu \bR_{ss}\right\|^{i-\ell} 
    \, \tau_m\,\Vert\bR_{tt,m}\Vert\,,
    \label{eq:stability_mean_1}
\end{align}
where $\|\bv\|_{{\bR_{tt,m}}}$ denotes $\bv^\top\bR_{tt,m}\bv$. The last inequality follows because all norms are equivalent in finite dimensional spaces, that is, there exists a $\tau_m>0$ such that:
\begin{equation}
    \label{eq:norm-equiv}
    \|\bv\|
    \leq
    \tau_m\,\|\bv\|_{{\bR_{tt,m}}}\,,
\end{equation}
for all $\bv$. The series $\sum_{\ell=0}^i \left\|\bI - \mu \bR_{ss}\right\|^{i-\ell} \sum_{m=1}^N\tau_m\,\|\bR_{tt,m}\|$ is absolutely convergent if the step-size $\mu$ is chosen according to \eqref{eq:stab_cond}.
The series $\sum_{\ell=0}^i\sum_{m=1}^N[(\bI-\mu\bR_{ss})^{i-\ell}\boldsymbol{q}_m(\ell)]_k$ is therefore absolutely convergent.

To show that the third series in~\eqref{eq:weight_error_closedForm}
converges, it is sufficient to observe that: 
\begin{equation}
    \label{eq:stability_mean_2}
    \begin{split}
            &\left|[(\bI - \mu \bR_{ss})^{i-\ell} \big(\br_{sy}-\bR_{ss}\bgamma^*\big)]_k\right| \\
            &\hspace{2cm}\leq \left\|\bI - \mu \bR_{ss}\right\|^{i-\ell} 
    \Vert\br_{sy}-\bR_{ss}\bgamma^*\Vert\,.
    \end{split}
\end{equation}
The series $\sum_{\ell=0}^i\left\|\bI - \mu\bR_{ss}\right\|^{i-\ell} \Vert\br_{sy}-\bR_{ss}\bgamma^*\Vert$ is absolutely convergent if the step-size $\mu$ is chosen according to \eqref{eq:stab_cond}.

To conclude, all series in \eqref{eq:weight_error_closedForm} converge if \eqref{eq:stab_cond} is satisfied. Therefore, for any initial condition, the algorithm~\eqref{eq:update} converges in the mean if the step-size $\mu$ is chosen according to condition~\eqref{eq:stab_cond}.

\subsection{Mean square error analysis}
\label{ss:mean_squared}
We examine now the mean-square-error behavior of the algorithm~\eqref{eq:update} by studying $\expec\{ \bv(i) \bv^\top\!(i)\}$.
The mean-square-error recursion is given by:
\begin{align}
    &\bV(i+1) 
    = \expec\{ \bv(i+1) \bv^\top\!(i+1)\} \nonumber \\
    %
    &= \expec\Bigg\{\sym\Bigg\{ 
    \frac{1}{2} \underbrace{\bv(i) \bv^\top\!(i)}_{\bQ_1}
    - \mu \underbrace{\bv(i)\bv^\top\!(i) \bs(i)\bs^\top\!(i)}_{\bQ_2} \nonumber \\
    & + \mu \underbrace{\bv(i)\bs^\top\!(i) e_{0}(i)}_{\bQ_3}
    - \mu^2 \underbrace{\bs(i)\bs^\top\!(i) \bv(i) \bs^\top\!(i) e_{0}(i)}_{\bQ_4} 
    \nonumber \\
    & + \frac{\mu^2}{2} \underbrace{\bs(i) \bs^\top\!(i) (e_{0}(i))^2}_{\bQ_5} + \frac{\mu^2}{2} \underbrace{\bs(i)\bs^\top\!(i) \bv(i) \bv^\top\!(i) \bs(i)\bs^\top\!(i)}_{\bQ_6}\nonumber \\
    & - \mu \eta \underbrace{\bv(i)\bigg[\sum_{m=1}^N \dfrac{\hat\bR_{tt,m}(i) \hat{\bgamma}(i)}{\Delta_m(i)}\bigg]^\top}_{\bQ_7} \nonumber \\
    & + \mu^2 \eta \underbrace{\bs(i)\bs^\top\!(i) \bv(i) \bigg[\sum_{m=1}^N \dfrac{\hat\bR_{tt,m}(i) \hat{\bgamma}(i)}{\Delta_m(i)}  \bigg]^\top}_{\bQ_8}
    \nonumber \\ 
    & - \mu^2 \eta \underbrace{\bs(i) e_{0}(i) \bigg[\sum_{m=1}^N \dfrac{\hat\bR_{tt,m}(i) \hat{\bgamma}(i)}{\Delta_m(i)}  \bigg]^\top}_{\bQ_9}
    \nonumber \\
    & + \frac{\mu^2\eta^2}{2} \underbrace{\bigg[\sum_{m=1}^N \dfrac{\hat\bR_{tt,m}(i) \hat{\bgamma}(i)}{\Delta_m(i)} \bigg] \bigg[\sum_{p=1}^N \dfrac{\hat\bR_{tt,p}(i) \hat{\bgamma}(i)}{\Delta_p(i)}  \bigg]^\top}_{\bQ_{10}}\Bigg\}\Bigg\}\,,
    \label{eq:Vi_recursion}
\end{align}
where $\sym(\bX)=\bX+\bX^\top$. We shall now calculate the expectations separately. Helpful details are provided in \cite{parreira2012stochasticGaussianKLMS}:

\begin{itemize}
    \item[$\bQ_1$]   
    $\expec\{\bv(i)\bv^\top\!(i)\}=\bV(i)$ appears on the r.h.s. of~\eqref{eq:Vi_recursion}.
    
    \item[$\bQ_2$] Using some algebraic manipulations and assumption \textsf{A4}:
    \begin{align}
        \expec\{\bv(i) \bv^\top\!(i) \bs(i)\bs^\top\!(i)\}
        & = \expec\{\bv(i) \bv^\top\!(i) \}  \expec\{ \bs(i)\bs^\top\!(i)\}
        \nonumber \\
        & = \bV(i) \bR_{ss}\,,
        \label{eq:term2}
    \end{align}
    where $\bR_{ss}$ is calculated in closed-form in the supplementary material for this paper.
    
    \item[$\bQ_3$] Using assumption \textsf{A4}, we obtain:
    \begin{align}
        \expec\{\bv(i) \bs^\top\!(i)  e_{0}(i) \} 
        & = \expec\{\bv(i) \} \expec\{ \bs^\top\!(i) e_{0}(i) \} 
        \nonumber \\
        & = \expec\{\bv(i) \}\left( \br_{sy}
        - \bR_{ss}\bgamma^* \right)^\top \,,
        \label{eq:term3}
    \end{align}
    where all these quantities have already been calculated.
    
    \item[$\bQ_4$] Considering hypothesis \textsf{A4} and assuming the input signal to be i.i.d., we have: 
    \begin{align}
        \big[&\expec \{\bs(i)\bs^\top\!(i) \bv(i) \bs^\top\!(i) e_{0}(i) \} \big]_{u,v} \nonumber \\
        &= \sum_{a=1}^{N|\cD_n|} \expec\{ s_u(i) s_a(i) v_a(i) s_v(i) e_{0}(i) \}
        \nonumber \\
        &= \sum_{a=1}^{N|\cD_n|} \expec\{v_a(i)\} \expec\{ s_u(i) s_a(i) s_v(i) [y_n(i)-\bs^\top\!(i)\bgamma^*] \}
        \nonumber \\
        &= - \sum_{b=1}^{N|\cD_n|} \sum_{a=1}^{N|\cD_n|} \expec\{v_a(i)\} \expec\{ s_u(i) s_a(i) s_v(i) s_b(i) \} \gamma_b^*
        \nonumber \\
        &\quad + \sum_{a=1}^{N|\cD_n|} \expec\{v_a(i)\} \expec\{ s_u(i) s_a(i) s_v(i) y_n(i) \}\,,
        \label{eq:term4}
    \end{align}
    where both expectations $\expec\{ s_u(i) s_a(i) s_v(i) y_n(i) \}$ and $\expec\{ s_u(i) s_a(i) s_v(i) s_b(i)\}$ are provided in closed-form in the supplementary material for this paper.
    
    \item[$\bQ_5$] 
    This term can be written as:
    \begin{align}
        &\big[\expec \{\bs(i) \bs^\top\!(i) e_{0}^2(i)\} \big]_{u,v} \nonumber \\
        &= \expec\{s_u(i) s_v(i) [y_n(i)-\bs^\top\!(i)\bgamma^*]^2\}
        \nonumber \\
        &= \expec\{s_u(i) s_v(i) y_n^2(i) \}\nonumber \\
        &\quad- 2 \sum_{p} \gamma^*_p \expec\{s_u(i) s_v(i) s_p(i) y_n(i) \}
        \nonumber \\ & 
        \quad+ \sum_{\ell}\sum_m \gamma^*_{\ell} \gamma^*_m \expec\{s_u(i) s_v(i) s_{\ell}(i) s_m(i)\}\,.
        \label{eq:term5}
    \end{align}

    \item[$\bQ_6$] We use the strategy in \cite{parreira2012stochasticGaussianKLMS} and assumption \textsf{A4} in order to obtain the following approximation:
    \begin{align}
    &\big[\expec \{\bs(i)\bs^\top\!(i) \bv(i) \bv^\top\!(i) \bs(i)\bs^\top\!(i) \} \big]_{u,v} \nonumber \\
    & = \sum_{\ell} \sum_{m} \expec\{ s_u(i) s_{\ell}(i) v_{\ell}(i) v_m(i)s_m(i) s_v(i)\} \nonumber
    \\
    & \approx \sum_{\ell} \sum_{m} \expec\{s_u(i) s_{\ell}(i) s_m(i) s_v(i)\} \expec\{v_{\ell}(i) v_m(i)\}\,,
    \label{eq:suxyv}
    \end{align}
    where $\expec\{v_{\ell}(i) v_m(i)\}=\big[\bV(i)\big]_{\ell,m}$. A closed-form expression of the expectation $\expec\{s_u(i) s_{\ell}(i) s_m(i) s_v(i)\}$ is provided in the supplementary material for this paper.
    
    \item[$\bQ_7$] Recalling that $\bv(i)=\hat{\bgamma}(i)-\bgamma^*$, this term becomes:
    \begin{align}
    \expec & \Bigg\{
    \bv(i)
    \Bigg[\sum_{m=1}^N \dfrac{\hat\bR_{tt,m}(i) \hat{\bgamma}(i)}{\Delta_m(i)}  \Bigg]^\top 
    \Bigg\}
    \nonumber\\
    %
    &= \sum_{m=1}^N \expec\left\{ \dfrac{\hat{\bgamma}(i)\hat{\bgamma}^\top\!(i)\hat\bR_{tt,m}(i)}{{\Delta_m(i)}}
    \right\} \nonumber \\
    & - \sum_{m=1}^N \bgamma^* \expec\left\{ \dfrac{\hat{\bgamma}^\top\!(i)\hat\bR_{tt,m}(i)}{\Delta_m(i)}
    \right\}\,.
    \label{eq:term7_aux}
    \end{align}
    The second term on the r.h.s. of \eqref{eq:term7_aux} can be calculated with \eqref{eq:approx_reg}. To calculate the first term on the r.h.s. of \eqref{eq:term7_aux}, we use the following approximation:
    \end{itemize}
   \begin{align}
    \expec & \left\{ \dfrac{\hat{\bgamma}(i)\hat{\bgamma}^\top\!(i)\hat\bR_{tt,m}(i)}{{\Delta_m(i)}}
    \right\} 
    \approx 
    \expec\left\{ \dfrac{\hat{\bgamma}(i)\hat{\bgamma}^\top\!(i)\bR_{tt,m}}{{\Delta_m(i)}} \right\}
    \nonumber \\
    & \approx
    \dfrac{\left[\bSig(i) +\bmu(i)\,\bmu^\top\!(i)\right]\bR_{tt,m}}{\sqrt{\tr\{\bR_{tt,m}\bSig(i)\}+\bmu^\top\!(i)\bR_{tt,m}\bmu(i)}} \,,
    \label{eq:term7}
    \end{align}
    \begin{itemize}
    
    \item[$\bQ_8$] By using some algebraic manipulations and \textsf{A4}, we get:
    \begin{align}
        \expec & \left\{\bs(i)\bs^\top\!(i) \bv(i) \bigg[\sum_{m=1}^N \dfrac{\hat\bR_{tt,m}(i) \hat{\bgamma}(i)}{{\Delta_m(i)}}  \bigg]^\top\right\}
        \nonumber \\
        & = \expec\big\{\bs(i)\bs^\top\!(i)\big\} \expec\bigg\{\bv(i) \bigg[\sum_{m=1}^N \dfrac{\hat\bR_{tt,m}(i) \hat{\bgamma}(i)}{{\Delta_m(i)}}\bigg]^\top \bigg\}
        \nonumber \\ 
        & = \bR_{ss} \bQ_7\,.
        \label{eq:term8}
    \end{align}
    \item[$\bQ_9$] By approximating $\hat\bR_{tt,m}(i)\approx \bR_{tt,m} $ and using assumption \textsf{A4}, we obtain:
    \begin{align}
        \expec & \left\{\bs(i) e_{0}(i) \bigg[\sum_{m=1}^N \dfrac{\hat\bR_{tt,m}(i) \hat{\bgamma}(i)}{\Delta_m(i)}  \bigg]^\top\right\}
        \nonumber \\
        & = \expec\big\{\bs(i) e_{0}(i)\big\}\,\expec\bigg\{ \sum_{m=1}^N \dfrac{\hat\bR_{tt,m}(i) \hat{\bgamma}(i)}{\sqrt{\hat{\bgamma}^\top\!(i) \hat\bR_{tt,m}(i)\hat{\bgamma}(i)}}\bigg\}^\top
        \nonumber \\ 
        &= \left(\br_{sy} - \bR_{ss}\bgamma^*\right)
        \label{eq:term9} \\
        &\quad\times\sum_{m=1}^N\dfrac{\bR_{tt,m}\bmu(i)}{\sqrt{\tr\{\bR_{tt,m}\bSig(i)\}+\bmu^\top\!(i)\bR_{tt,m}\bmu(i)}}\nonumber.
    \end{align}
    
    \item[$\bQ_{10}$] We successively have:
    \end{itemize}
    \vspace{-0.2cm}
    \begin{align}
        & \expec \left\{\bigg[ \sum_{m=1}^N \dfrac{\hat\bR_{tt,m}(i) \hat{\bgamma}(i)}{{\Delta}_m(i)} \bigg]\bigg[\sum_{p=1}^N \dfrac{\hat\bR_{tt,p}(i) \hat{\bgamma}(i)}{{\Delta}_p(i)}  \bigg]^\top 
        \right\} 
        \nonumber \\
        \approx & \sum_{m=1}^N \sum_{p=1}^N \expec\Bigg\{  \dfrac{\bR_{tt,m} \hat{\bgamma}(i)\hat{\bgamma}^\top\!(i)\bR_{tt,p}}{\sqrt{\hat{\bgamma}^\top\!(i) \bR_{tt,m}\hat{\bgamma}(i) \hat{\bgamma}^\top\!(i) \bR_{tt,p}\hat{\bgamma}(i)}} \Bigg\}
        \nonumber \\
        \approx & \sum_{m=1}^N \sum_{p=1}^N  \dfrac{\bR_{tt,m} \left[\bSig(i) +\bmu(i)\,\bmu^\top\!(i)\right]\big\} \bR_{tt,p}}{\sqrt{\expec\big\{\hat{\bgamma}^\top\!(i) \bR_{tt,m}\hat{\bgamma}(i)\hat{\bgamma}^\top\!(i) \bR_{tt,p}\hat{\bgamma}(i)\big\}}}\,.
        \label{eq:term10}
    \end{align}
    \begin{itemize}
    \item[] Considering hypothesis \textsf{A6}, the denominator of this expression consists of the expectation of a linear combination of fourth order moments of Gaussian random variables. Using~\cite{kumar1973expectation}, we have:
    \begin{align}
        &\expec \big\{\hat{\bgamma}^\top\!(i) \bR_{tt,m}\hat{\bgamma}(i)\hat{\bgamma}^\top\!(i) \bR_{tt,p}\hat{\bgamma}(i)\big\} \nonumber \\
        &= 4\bmu^\top \bR_{tt,m} \bSig \bR_{tt,p} \bmu  
        + 2\tr\big\{\bR_{tt,m} \bSig \bR_{tt,p} \bSig \big\}
        \nonumber \\ 
        & + \Big(\bmu^\top \bR_{tt,m} \bmu + \tr\big\{\bR_{tt,m}\bSig\big\}\Big) \nonumber \\ & \times
        \Big(\bmu^\top \bR_{tt,p} \bmu + \tr\big\{\bR_{tt,p}\bSig\big\}\Big)\,,
    \end{align}
\end{itemize}

To conclude, using the results from relations \eqref{eq:term2}--\eqref{eq:term10}, we arrive at the following recursive equation for $\bV(i)$:
\begin{align}
    \label{eq:bV}
    \bV( & i+1) = \bV(i)
    - \mu \sym\{\bV(i)\bR_{ss}\} 
    + \mu\sym\{\bQ_3\} 
    \nonumber \\ &
    - \mu^2\sym\{\bQ_4\}
    + \mu^2 \bQ_6
    - \mu\eta \sym\{\bQ_7\} + \mu^2\eta^2 \bQ_{10}
    \nonumber \\ & 
    + \mu^2\bQ_5 + \mu^2\eta \sym\{\bR_{ss}\bQ_7\}
    - \mu^2\eta \sym\{\bQ_9\}\,.
\end{align}
The unknown moments are calculated in closed-form in the supplementary material for this paper.


Before going further, we define the Mean Square Deviation:
\begin{align}
	\text{MSD}(i) & \triangleq \expec\left\{\left\Vert  \hat{\bgamma}(i) - \bgamma^* \right\Vert^2\right\} 
	\nonumber \\
	& = \expec\left\{\left\Vert  \bv(i) \right\Vert^2\right\} = \tr\left\{\bV(i)\right\}\,.
\label{eq:msd_def}
\end{align}
This quantity is useful in determining the performance of the algorithm.

\subsection{Steady-state performance when \texorpdfstring{$\eta=0$}{n = 0}}

Assuming a small enough step-size $\mu$ that verifies \eqref{eq:stab_cond}, with relation \eqref{eq:weight_error_closedForm} we have:
\begin{equation}
    \begin{split}
	\bv(\infty) &= \lim_{i\to \infty}\expec\left \{\bv(i)\right\} \\
	&= \Zero\,.
	\end{split}
\end{equation}
%
%
%
%
We shall now evaluate $\bV(\infty)=\lim_{i\to \infty}\bV(i)$ using \eqref{eq:bV}. With $\bv(\infty)=\Zero$, we know that $\bQ_3(\infty) = \bQ_4(\infty) = \Zero$. Using vectorization of \eqref{eq:bV} together with the Kronecker product gives:
\begin{equation}
    \vc\{\bV(\infty)\} = 
    \mu^2 \left(\bI_2 - \boldF_0\right)^{-1}\vc\{\bQ_5\}\,,
\end{equation}
with:
\begin{equation}
    \boldF_0 = \bI_2 - \mu (\bI \otimes \bR_{ss} + \bR_{ss} \otimes \bI) + \mu^2\boldF_1\,,
\label{eq:F0_5}
\end{equation}
where operator $\otimes$ is the Kronecker product.
In these expressions, $\bI_2$ is the identity matrix of size $k_s^2 \times k_s^2$, and $\bI$ is the identity matrix of size $k_s \times k_s$, where $k_s = (N+1)\vert \cD_n \vert$ is the number of entries in block vector $\bs$. The entries of $\boldF_1$ are:
\begin{equation}
    \left[\boldF_1\right]_{u+(\ell-1)k_s,m+(v-1)k_s}=\expec\{s_u(i) s_{\ell}(i) s_m(i) s_v(i)\}\,.
\end{equation}
They are provided in closed-form in the supplementary material for this paper.

Finally, we conclude that:
\begin{equation}
	\text{MSD}(\infty) = \tr\left\{\vc^{-1}\left\{\mu^2 \left(\bI_2 - \boldF_0\right)^{-1}\vc\{\bQ_5\}\right\}\right\}\,.
\end{equation}


\section{Experimental validation}
\label{sect:experimental}

The performance of algorithm~\eqref{eq:update} has already been illustrated in the preliminary work~\cite{moscu2020online}, both with simulated and real-world data. Consequently, the experiments in this section will be restricted to model validation.

\subsection{Signal model}
\label{subsub:eta_zero}

We considered a simulation scenario with i.i.d. zero-mean Gaussian data $\by(i)$ with covariance matrix $\bR_{yy}$ satisfying the following model:
\begin{equation}
    \by(i) = \bA \by(i) + \bv(i)\,,
    \label{eq:linear_model_1}
\end{equation}
where $\bA$ is the adjacency matrix defined as:
\begin{equation}
    \label{eq:adj_mat}
    \bA = \begin{pmatrix}
    0 & 1 & 0 & 1 & 1\\
    1 & 0 & 1 & 0 & 1\\
    1 & 0 & 0 & 1 & 0\\
    0 & 1 & 1 & 0 & 1\\
    1 & 0 & 1 & 1 & 0
    \end{pmatrix}\,.
\end{equation}
and $\bv(i)$ is a zero-mean white Gaussian noise with covariance matrix $\sigma^2_v\bI_5$ and $\sigma_{v}=0.05$. Note that $\bR_{yy}$ satisfies:
\begin{equation}
    \bR_{yy} = [\bI-\bA]^{-1}\bR_{vv}([\bI-\bA]^{-1})^\top\,.
\end{equation}
This model, although not nonlinear, offers exact knowledge of the statistical properties necessary to validate the models derived in this paper. Validation on a purely nonlinear system follows in subsection \ref{subsec:lin_valid}. The reader is also referred to~\cite{moscu2020online} to assess the methods' behavior in real-world nonlinear settings. 

We used our algorithm with the Gaussian kernel~\eqref{eq:gaussian_ker}, with kernel bandwidth $\sigma = 1$. Each node stored a dictionary~$\cD_n$ with $6$ elements chosen uniformly in $[-1,1]^4$. Simulations were averaged over 100 Monte-Carlo runs. Results presented hereafter focus on node $n=1$.

\subsection{Computing the optimal coefficients}


In order to compute the optimal coefficients, we note that problem~\eqref{eq:opt_problem_online2} can be written equivalently as:
\begin{align}
    & \min_{\bgamma} \,\,
    \dfrac{1}{2} \bgamma^\top \bR_{ss} \bgamma
    - \bgamma^\top \br_{sy}  + \eta \sum_{m=1}^N \|\bC_m^\top\bgamma\|_2\,,
\end{align}
where $\bC_m\bC_m^\top=\bR_{tt,m}$.
We solved this optimization problem using the CVX package to get $\bgamma^*$ \cite{cvx,gb08}.

The MSD was evaluated experimentally using:
\begin{equation}
\text{MSD}(i) = \expec\left\{\left\Vert \hat{\bgamma}(i) - \bgamma^* \right\Vert^2 \right\}\,.
\label{eq:exp_msd_exp}
\end{equation}

\subsection{Parameter \texorpdfstring{$\eta = 0$}{n=0}}

Figure~\ref{subf:ME_theo5} shows the theoretical and experimental learning curves of the entries of the parameter vector $\bgamma$.
Figure~\ref{subf:MSD_theo5} shows both the theoretical and experimental MSD curves, as well as the steady-state MSD. The theoretical curves closely follow the theoretical ones, both in the analyses in the mean and mean square sense. They also validate the approximation \eqref{eq:approx_reg} of the regularization term which is in the update \eqref{eq:expec_weight_error}.

\begin{figure}[t]
	\subfigure[Theoretical and experimental entries of $\bgamma$. Blue dashed lines are the experimental curves, while red lines are the theoretical ones]{\includegraphics[width=1.00\columnwidth,clip, trim=1.6cm 6.2cm 2.2cm 7.1cm]{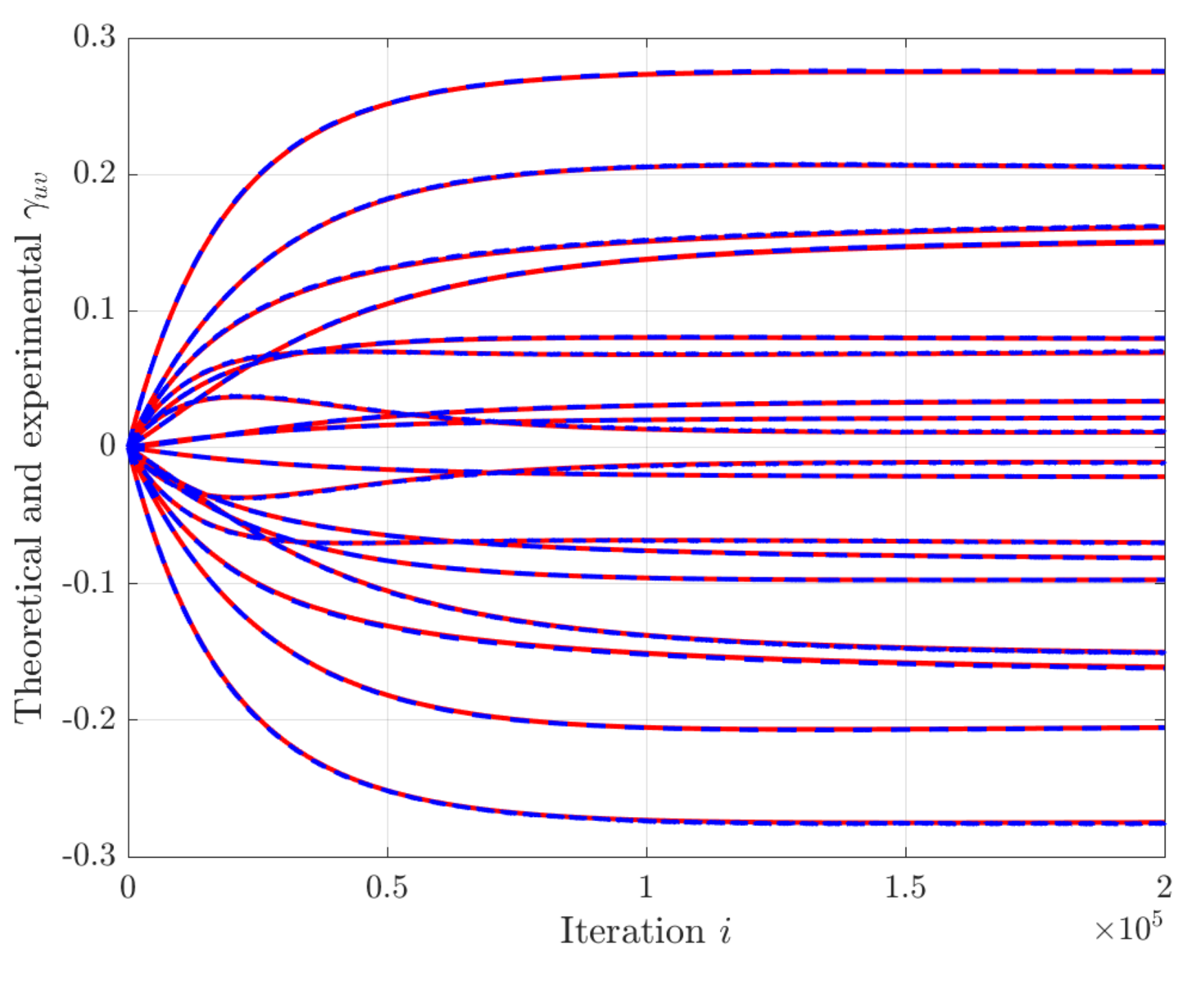}
		\label{subf:ME_theo5}}
	\hfill
	\subfigure[Experimental, steady-state, and theoretical MSD]{\includegraphics[width=1.00\columnwidth,clip, trim=1.6cm 6.2cm 2.2cm 7.1cm]{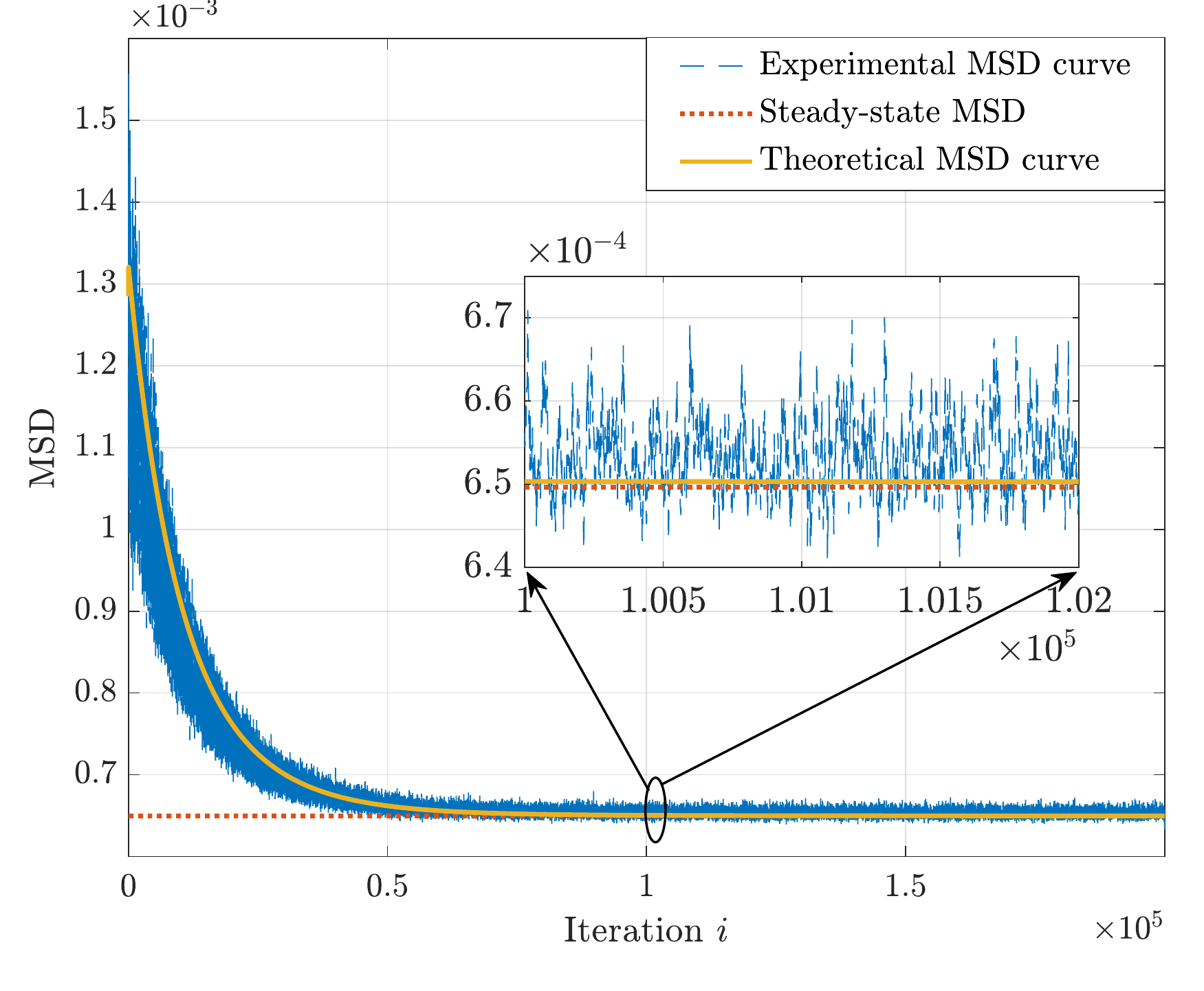}
		\label{subf:MSD_theo5}}
	\caption{Analysis validation in mean and mean square sense, for $\eta = 0$}
	\label{fig:me_mse5}
\end{figure}

\subsection{Parameter \texorpdfstring{$\eta > 0$}{n>0}}

The conditions of this experiment were the same as for the previous case, with $\eta = 1\cdot10^{-4}$.
%

The learning curves are presented in Fig.~\ref{fig:me_mse_reg5}.
As for the previous case, we observe that the theoretical curves closely fit the experimental ones. This validates the assumptions and approximations employed in the analysis to make it tractable.

\begin{figure}[t]
	\subfigure[Theoretical and experimental entries of $\bgamma$. Blue dashed lines are the experimental curves, while red lines are the theoretical ones]{\includegraphics[width=1.00\columnwidth,clip, trim=1.6cm 6.2cm 2.2cm 7.1cm]{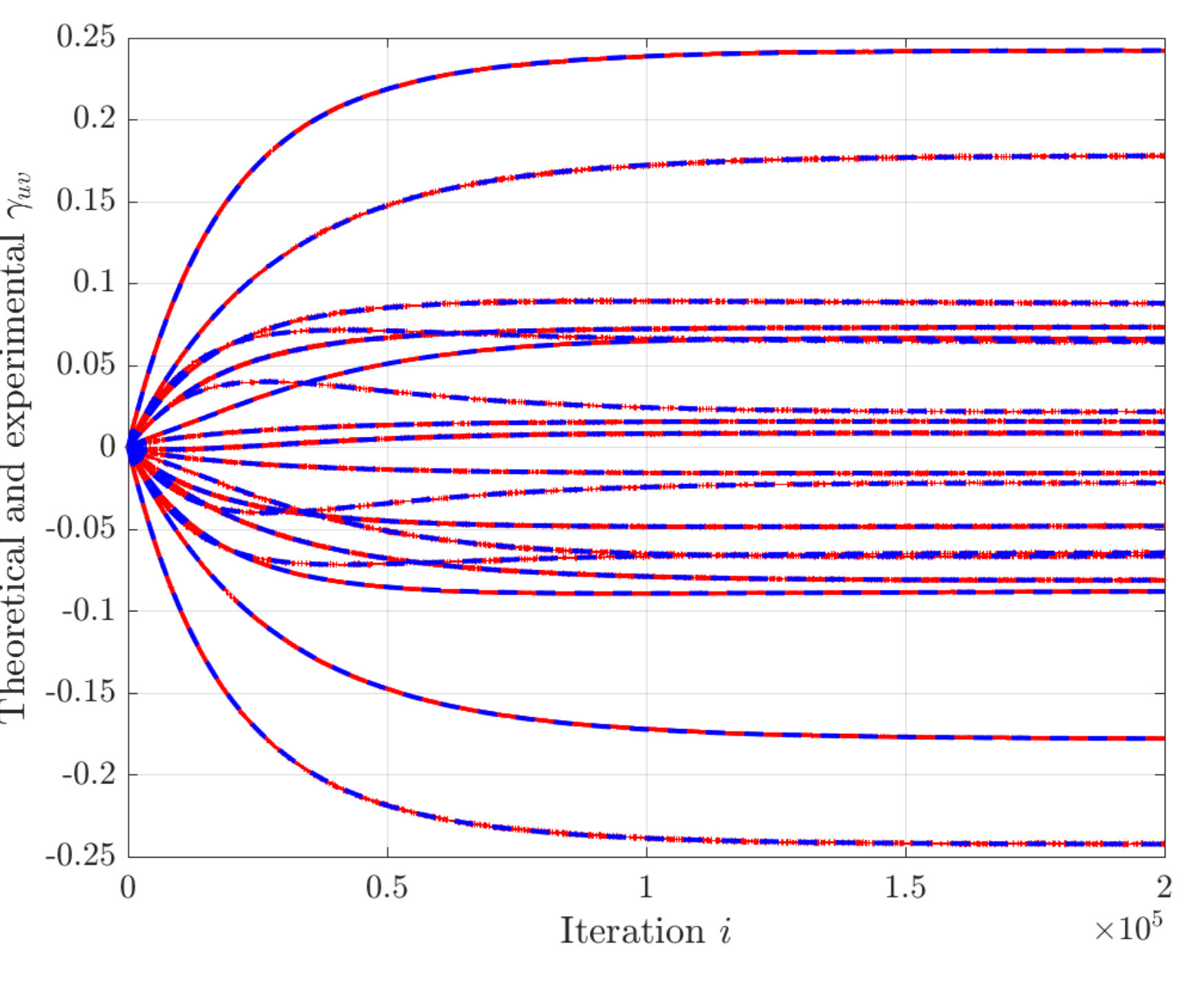}
		\label{subf:ME_reg_theo5}}
	\hfill
	\subfigure[Experimental and theoretical MSD]{\includegraphics[width=1.00\columnwidth,clip, trim=1.6cm 6.2cm 2.2cm 7.1cm]{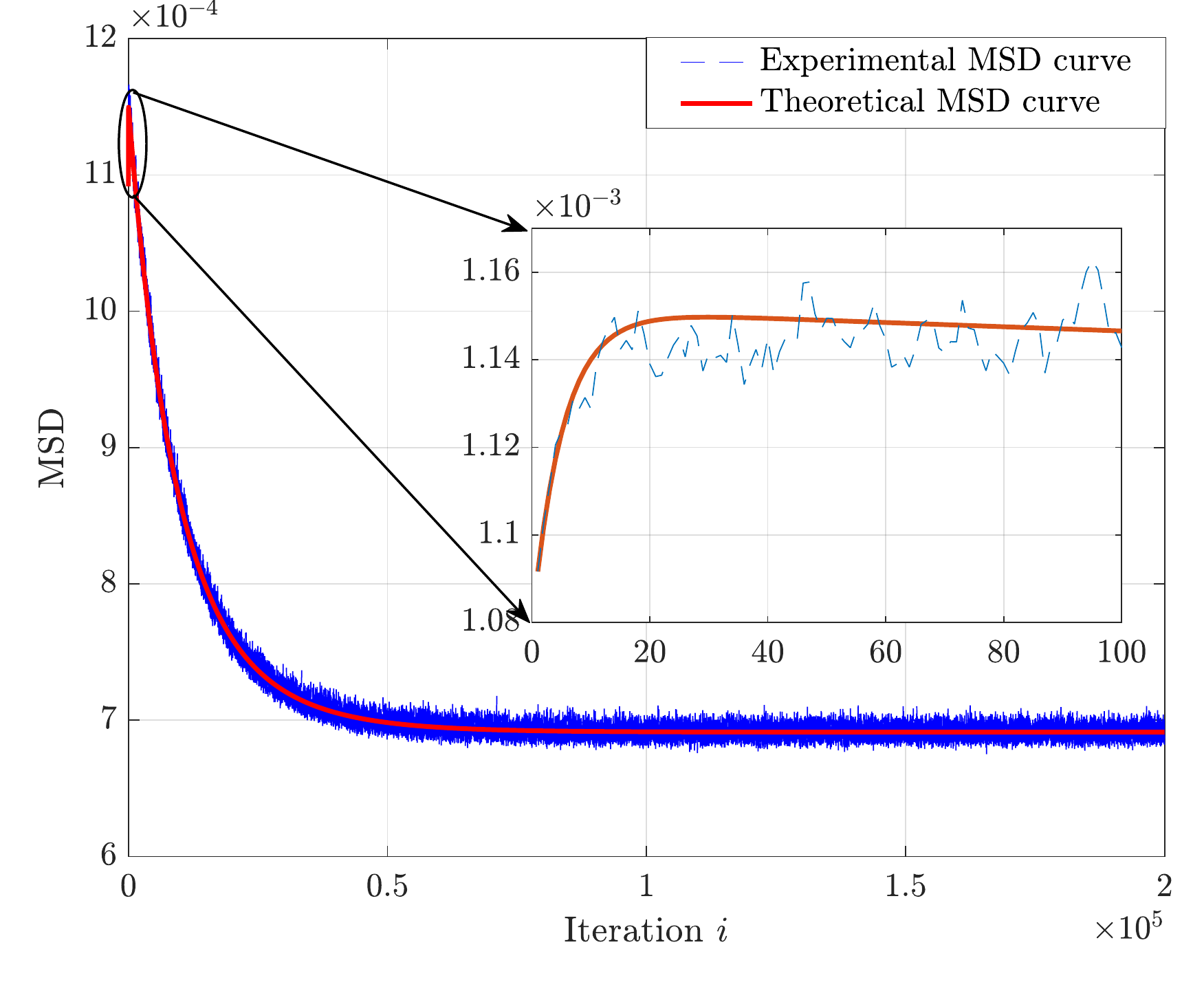}
		\label{subf:MSD_reg_theo5}}
	\caption{Analysis validation in mean and mean square sense, for $\eta = 1\cdot10^{-4}$}
	\label{fig:me_mse_reg5}
\end{figure}

\subsection{Validation with a nonlinear model}
\label{subsec:lin_valid}

In order to evaluate how well the model can predict the the behavior of the algorithm in nonlinear settings, we consider an additional example with a 3-node graph with adjacency matrix given by:
\begin{align}
    \bA = \begin{pmatrix}
    0 & 0 & 1\\
    1 & 0 & 1\\
    1 & 0 & 0\\
    \end{pmatrix},
\end{align}
and i.i.d. streaming nodal data generating to agree with the following model:
\begin{align}
    \by(i) &= \underbrace{\begin{bmatrix}
    y_1 - K_1 (y_3+y_1)^3 (K_2 y_1)^{-1} \\
    %
    y_2 + \frac{y_2-K_1 (y_3+y_1)^3 (K_2 y_1)^{-1}}{[0.5+\exp(K_1 (y_3+y_1)^3 (K_2 y_1)^{-1})]^5+1} \\
    y_3 + y_1 + K_1 (y_3+y_1)^3 (K_2 y_1)^{-1}
    \end{bmatrix}}_{\boldf(\by(i))}
    + \boldsymbol{\varrho}(i)
    \nonumber
\end{align}
with $K_1=8000$ and $K_2=27$, $\boldsymbol{\varrho}(i)$ being an i.i.d. zero-mean Gaussian random vector with covariance matrix~$\bI_3$, and $\boldf(\by(i))=\big[f_1(\by(i)),f_2(\by(i)),f_3(\by(i))\big]^\top$ the vector-valued function corresponding to the concatenation of $f_n$ in~\eqref{eq:model_central}. Notice that the sample index $i$ was omitted inside $\boldf(\by(i))$ to simplify the notation. Graph signal $\by(i)$ can be sampled from this model by first sampling from $\boldsymbol{\varrho}(i)$ and using the relation $\by(i)=(\operatorname{id}-\boldf)^{-1}(\boldsymbol{\varrho}(i))$, with $\operatorname{id}(\bx)$ the identity function.

We used our algorithm with the Gaussian kernel~\eqref{eq:gaussian_ker}, with kernel bandwidth $\sigma = 1$. Each node stored a dictionary~$\cD_n$ with $4$ elements chosen uniformly in $[-1,1]\times[-1,1]$. Simulations were averaged over 100 Monte-Carlo runs. Results presented hereafter focus on node $n=1$.

Figure~\ref{fig:me_mse5_nonlin} depicts the behavior of the algorithm in mean and mean-square sense without regularization (i.e., for $\eta=0$). Figure~\ref{fig:me_mse_reg5_nonlin} shows the mean and mean-square performance in the case when $\eta = 0.3$. These two figures together show that the conducted analysis is indeed able to predict the behavior of the algorithm without the need of running multiple \mbox{Monte-Carlo} runs. Moreover, even with the highly nonlinear nature of $\boldf$ and $\by(i)$ not being Gaussian distributed, the theoretical curves still follow the experimental ones relatively closely, although a small amount of bias can be observed in steady state.

\begin{figure}[t]
	\subfigure[Theoretical and experimental entries of $\bgamma$. Blue dashed lines are the experimental curves, while red lines are the theoretical ones]{\includegraphics[width=0.98\columnwidth,clip, trim=3.9cm 8.4cm 4.3cm 9.0cm]{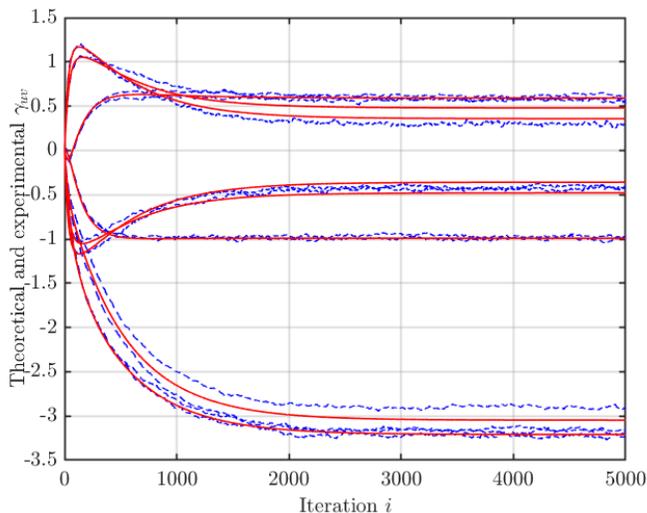}
		\label{subf:ME_theo5_nonlin}}
	\hfill
	\subfigure[Experimental, steady-state, and theoretical MSD]{\includegraphics[width=1.00\columnwidth,clip, trim=3.9cm 8.4cm 4.3cm 9.0cm]{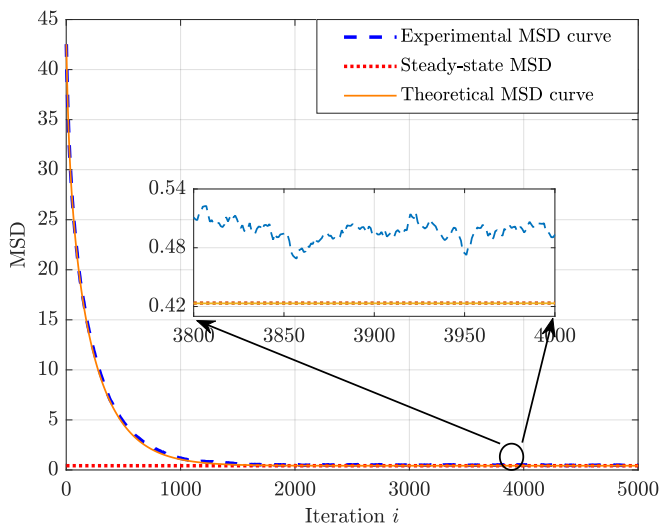}
		\label{subf:MSD_theo5_nonlin}}
	\caption{Analysis validation in mean and mean-square sense, for $\eta = 0$}
	\label{fig:me_mse5_nonlin}
\end{figure}

\begin{figure}[t]
	\subfigure[Theoretical and experimental entries of $\bgamma$. Blue dashed lines are the experimental curves, while red lines are the theoretical ones]{\includegraphics[width=1.00\columnwidth,clip, trim=3.6cm 8.4cm 4.3cm 9.0cm]{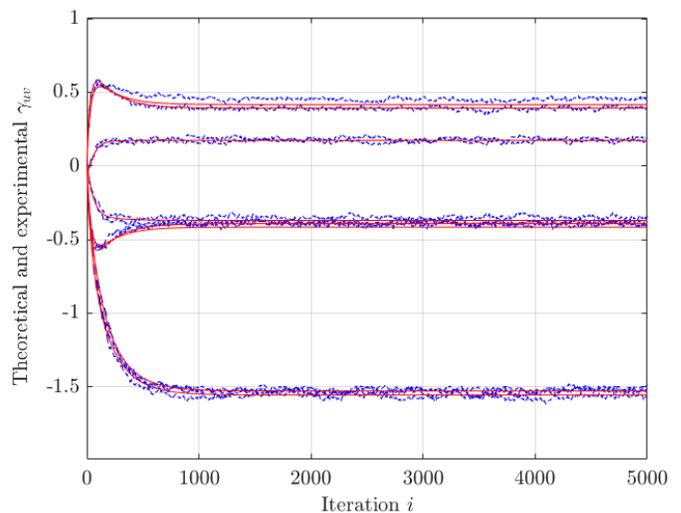}
		\label{subf:ME_reg_theo5_nonlin}}
	\hfill
	\subfigure[Experimental and theoretical MSD]{\includegraphics[width=1.00\columnwidth,clip, trim=3.9cm 8.4cm 4.3cm 9.0cm]{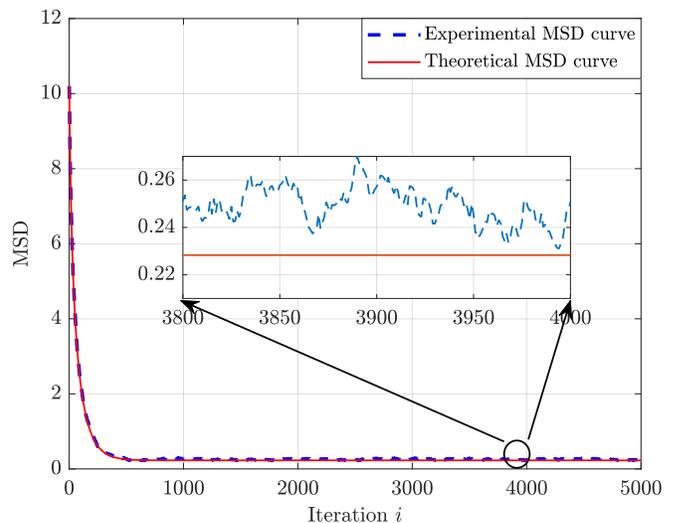}
		\label{subf:MSD_reg_theo5_nonlin}}
	\caption{Analysis validation in mean and mean-square sense, for $\eta = 0.3$}
	\label{fig:me_mse_reg5_nonlin}
\end{figure}

\section{Conclusion}
\label{sect:conclusion}

In this paper, an online graph-topology inference method recently proposed by the authors was recalled and analyzed. While previous works mainly focused on models based on additive interactions between the signals at each node, we considered arbitrary nonlinear interactions between the nodes, which render our model much more general. By encoding the links as partial derivatives of some nonlinear functions, our method benefits from the kernel {machinery} framework to estimate a possibly directed, sparse adjacency matrix. However, it also demonstrates significantly increased complexity that makes its analysis non-trivial. This work proposed a thorough analysis of the algorithm, as well as performance bounds and stability conditions. 
The derived models characterize the performance of the algorithm as a function of important parameters such as the kernel bandwidth or the step-size. In turn, this allows for precise tuning of such parameters in order to obtain a desired performance in transient or in steady state. Moreover, the behavior of the intuitive partial-derivative-based sparsity regularizer was analyzed in-depth, thus completely characterizing the method.
The experimental results, as well as the algorithm analysis, show that the proposed method can lead to more accurate estimates for more general nonlinear systems.

\bibliographystyle{IEEEtran}
\bibliography{references}

\appendices

\section{Proof of Theorem \ref{theorem1}}
\label{proof_theo1}

\begin{proof}
Using orthogonal projection, we can decompose any function $f_n \in\mathcal{H}_{\kappa}$ as the sum of two functions:
\begin{align} 
    f_n = f_n^{\|} + f_n^{\perp}\,,
    \label{eq:decomp_fn_representer_annex}
\end{align}
where $f_n^{\perp}$ is orthonormal to $f_n^{\|}$, that is, $\langle f_n^{\|},f_n^{\perp}\rangle_{\mathcal{H}}=0$), and $f_n^{\|}$ can be written as:
\begin{align}
    f_n^{\|} = \sum_{\ell=1}^i \alpha_{\ell} \kappa\left(\cdot,\by(\ell)\right)
    + \sum_{m=1}^N \sum_{p=1}^{i} \sum_{q=1}^{L_m} \beta_{m,p,q} \kappa_{\partial_{m,q}}\left(\cdot,\by(p)\right)\,,
    \label{eq:representer_theorem_parallel_annex}
\end{align}
where $\alpha_\ell$ and $\beta_{m,p,q}$ are coefficients. This means that $f_n^{\|}$ lies in the span of $\kappa\left(\cdot,\by(\ell)\right)$ and $\kappa_{\partial_{m,q}}\left(\cdot,\by(p)\right)$, from which the orthogonality condition implies that $\langle f_n^{\perp},\kappa\left(\cdot,\by(\ell)\right)\rangle_{\mathcal{H}}=0$ and $\langle f_n^{\perp},\kappa_{\partial_{m,q}}\left(\cdot,\by(p)\right)\rangle_{\mathcal{H}}=0$, for all $\ell,m,p,q$.

Assume that the kernel $\kappa(\cdot,\cdot)$ is at least twice differentiable. Then, the following relation holds~\cite{zhou2008derivativeKernels}:
\begin{align}
    \mathcal{H} \ni \dfrac{\partial f_{n}(\by)}{\partial y_{{m,q}} } 
    = \langle f_{n}, \kappa_{\partial_{m,q}}(\cdot,\by) \rangle_{\mathcal{H}}\,,
\end{align}
where $y_{{m,q}}$ is the $q$-th entry of $\by_{m}$.

Let us plug then decomposition~\eqref{eq:decomp_fn_representer_annex} of $f_n$ in problem~\eqref{eq:sparse_model}. For the first term, using the reproducing property, we have:
\begin{align}
    f_n(\by(\ell)) & = \langle f_n, \kappa(\cdot,\by(\ell)) \rangle_{\mathcal{H}}
    \\
    & = \langle  f_n^{\|} + f_n^{\perp}, \kappa(\cdot,\by(\ell)) \rangle_{\mathcal{H}}
    \\
    & = \langle  f_n^{\|}, \kappa(\cdot,\by(\ell)) \rangle_{\mathcal{H}}\,,
\end{align}
for all $\ell=1,\ldots,i$, where $\langle f_n^{\perp}, \kappa(\cdot,\by(\ell))\rangle_{\mathcal{H}}=0$ because $f_n^{\perp}$ is orthogonal to each term in~\eqref{eq:representer_theorem_parallel_annex}. For the second term:
\begin{align}
    \dfrac{\partial f_{n}(\by(p))}{\partial y_{{m,q}} } & = \dfrac{\partial f_{n}(\by)}{\partial y_{{m,q}} } \bigg|_{\by = \by(p)}
    \\
    & = \langle f_n^{\|} + f_n^{\perp}, \kappa_{\partial_{m,q}}(\cdot,\by(p)) \rangle_{\mathcal{H}}
    \\
    & = \langle f_n^{\|}, \kappa_{\partial_{m,q}}(\cdot,\by(p)) \rangle_{\mathcal{H}}\,,
\end{align}
for all $p=1,\ldots,i$, where $\langle f_n^{\perp}, \kappa_{\partial_{m,q}}(\cdot,\by(p)) \rangle_{\mathcal{H}}=0$ since $f_n^{\perp}$ is perpendicular to each term in~\eqref{eq:representer_theorem_parallel_annex}. For the last term, we have:
\begin{align}
    \psi\big(\|f_n\|_{\cH}\big) 
    & = \psi\big(\|f_n^{\|} + f_n^{\perp}\|_{\cH}\big)
    \\
    & = \psi\big(\|f_n^{\|}\|_{\cH} + \|f_n^{\perp}\|_{\cH}\big)\,.
\end{align}
Since $\|f_n^{\perp}\|_{\cH}$ does not influence the first two terms in problem~\eqref{eq:sparse_model}, if $\psi$ is monotonically increasing, then the solution is such that $\|f_n^{\perp}\|_{\cH}=0$.
\end{proof}

\section{Computing the terms in the optimization problem that depend on \texorpdfstring{$f_{n,\cD_n}$}{fD}}
\label{annex:fnd}

Let us recall the optimization problem \eqref{eq:opt_problem_online2_func}:
\begin{align}
    \min_{f_{n,\cD_n}} \,\,
    & \dfrac{1}{2} 
    \expec\Big\{\Big\|y_n(i) - f_{n,\cD_n}(\by(i)) \Big\|^2 \Big\}
    \nonumber \\&
    + \eta \sum_{m=1}^N \sqrt{\expec\bigg\{\bigg(\frac{\partial f_{n,\cD_n}(\by(i))}{\partial y_{m}(i)}\bigg)^2\bigg\}}\,.
\end{align}
We now evaluate the terms involving $f_{n,\cD_n}$, using the reproducing property and~\eqref{eq:der_rkhs}.
For the first term, we have:
%
%
\begin{align}
    f_{n,\cD_n}(\by(i)) &= \big\langle f_{n,\cD_n}, \kappa(\cdot,\by(i)) \big\rangle_{\cH}
    \\
    & = \sum_{p=1}^{\vert \cD_n \vert} \alpha_{p} \big\langle\kappa(\cdot,\by(\omega_p)),\kappa(\cdot,\by(i)) \big\rangle_{\cH} 
    \nonumber \\
    & + \sum_{m=1}^{N} \sum_{q=1}^{\vert \cD_n \vert} \beta_{m,q} \big\langle\kappa_{\partial_m}(\cdot,\by(\omega_q)), \, \kappa(\cdot,\by(i)) \big\rangle_{\cH}\,, \nonumber
\end{align}
from which:
\begin{subequations}
    \begin{empheq}[left={}]{align}
      & \big\langle \kappa(\cdot,\by(\omega_p)), \kappa(\cdot,\by(i)) \big\rangle_{\cH} = \kappa(\by(\omega_p),\by(i)) \label{subeq:f_a} \,, \\
      & \big\langle \kappa_{\partial_m}(\cdot,\by(\omega_q)), \kappa(\cdot,\by(i)) \big\rangle_{\cH} = \frac{\partial \kappa(\by(\omega_{q}),\by(i))}{\partial y_m(\omega_q)}  \label{subeq:f_b} \,.
    \end{empheq}
\end{subequations}
For the second term, we have:
\begin{align}
    &\frac{\partial f_{n,{\cD_n}}(\by(i))}{\partial y_{m_2}(i)} =  \big\langle f_{n,{\cD_n}}, \kappa_{\partial_{m_2}}(\cdot,\by(i)) \big\rangle_{\cH}
    \nonumber\\
    & = \sum_{p=1}^{\vert \cD_n \vert} \alpha_{p} \big\langle\kappa(\cdot,\by(\omega_p)),\kappa_{\partial_{m_2}}(\cdot,\by(i)) \big\rangle_{\cH} 
    \nonumber \\
    & + \sum_{m_2=1}^{N} \sum_{q=1}^{\vert \cD_n \vert} \beta_{m_2,q} \big\langle\kappa_{\partial_{m_2}}(\cdot,\by(\omega_q))\,,  \kappa_{\partial_{m_2}}(\cdot,\by(i)) \big\rangle_{\cH}\,,
\end{align}
from which:
\begin{subequations}
    \begin{empheq}[left={\hspace{-0.19cm}}]{align}
    & \big\langle \kappa(\cdot,\by(\omega_p)),  \kappa_{\partial_{m_2}}(\cdot,\by(i)) \big\rangle_{\cH} = \frac{\partial \kappa(\by(\omega_p),\by(i))}{\partial y_{m_2}(i)} \label{subeq:df_a} \,, \\
    & \big\langle \kappa_{\partial_{m_1}}(\cdot,\by(\omega_q)),  \kappa_{\partial_{m_2}}(\cdot,\by(i)) \big\rangle_{\cH} = \frac{\partial \kappa(\by(\omega_q),\by(i))}{\partial y_{m_1}(\omega_q)y_{m_2}(i)}  \label{subeq:df_b} \,.
    \end{empheq}
\end{subequations}

\section{Calculation of \texorpdfstring{$\bz_m(i)$, $\bzeta_m(i)$, and $\bl_{m_1,m_2}(i)$}{zm(i), zetam(i), and lm1m2(i)}}
\label{annex:gaussian}

With the Gaussian kernel $\kappa(\ba,\bb) = \exp \left( - {\Vert \ba - \bb \Vert^2}/{2\sigma^2} \right)$, quantities \eqref{eq:z}--\eqref{eq:l} are given by:
\begin{equation*}
    \begin{split}
    &[\bz_m(i)]_q = \kappa(\by(i),\by(\omega_q)) \left(\dfrac{y_m(i) - y_m(\omega_q)}{\sigma^2} \right) \\
    &[\bzeta_m(i)]_q = -[\bz_m(i)]_q
    \end{split}\,,
\end{equation*}

\begin{align}
    &[\bl_{m_1,m_2}(i)]_q =
    \nonumber \\ &
    \begin{cases}
    - \kappa(\by(i),\by(\omega_q)) \dfrac{y_{m_1}(i) - y_{m_1}(\omega_q)}{\sigma^2} \dfrac{y_{m_2}(i) - y_{m_2}(\omega_q)}{\sigma^2} &, m_1 \neq m_2  \\
    - \kappa(\by(i),\by(\omega_q)) \left( \dfrac{\big(y_{m_1}(i) - y_{m_1}(\omega_q) \big)^2}{\sigma^4} - \dfrac{1}{\sigma^2} \right) &, m_1 = m_2
\end{cases}\nonumber\,.
\end{align}

\iftrue

\clearpage

\onecolumn

\setcounter{page}{1}
\setcounter{equation}{0}

\title{Graph topology inference with derivative-reproducing property in RKHS: \\ algorithm and convergence analysis}
\maketitle

\begin{center}
    {\Large
    {Supplementary material for the paper \vspace{5mm} \\\bf{Graph topology inference with derivative-reproducing property in RKHS: \\ algorithm and convergence analysis}}}
\end{center}

\bigskip

The goal of the supplementary material is to provide detailed computations of quantities $\bR_{ss}$, $\bR_{tt,m}$, $\expec\{s_u(i)s_v(i)s_a(i)s_b(i)\}$, $\expec\{s_u(i)s_v(i)s_a(i)y_n(i)\}$, $\expec\{s_u(i)s_v(i)y_n^2(i)\}$ and $\expec\{s_u(i)y_n(i)\}$.


We introduce vector $\widetilde{\by}(i)=[\by^\top\!(i),\,y_n^\top\!(i)]^\top$, whose $k$-th entry is denoted by  $\widetilde{y}_k(i)$. The expectations involved in computing $\bR_{tt,m}$ (with $\bR_{ss}$ being one of its particular cases), $\expec\{s_u(i)s_v(i)s_a(i)s_b(i)\}$, $\expec\{s_u(i) s_v(i) s_a(i) y_n(i)\}$, $\expec\{s_u(i) s_v(i) y_n^2(i) \}$ and $\expec\{s_u(i)y_n(i)\}$ can be expressed generically as\footnote{Note that the $i$ index was omitted in~\eqref{eq:generic_form} to simplify the notation.}:
\begin{align}
    \expec\left\{f(\widetilde{\by}) \triangleq
    (\widetilde{y}_{h_1}-x_{h_2})^{\iota_1} 
    (\widetilde{y}_{h_3}-x_{h_4})^{\iota_2} 
    (\widetilde{y}_{h_5}-x_{h_6})^{\iota_3} 
    (\widetilde{y}_{h_7}-x_{h_8})^{\iota_4}
    \exp\left( -\dfrac{1}{2}c_4 \widetilde{\by}^\top\bB_0\widetilde{\by} - \bx_4^\top \bB_4^\top \widetilde{\by} + \dfrac{1}{2}\bx_4^\top\bQ_4\bx_4 \right) \right\}\,,
    \label{eq:generic_form}
\end{align}
where $\bx_4=[\by^\top\!(\omega_p),\by^\top\!(\omega_q),\by^\top\!(\omega_r),\by^\top\!(\omega_s)]^\top$ represents the fixed dictionary elements. 
Also, 

\begin{align}
    & \bB_0 = \begin{bmatrix} \bI & \Zero \\ \Zero & 0\end{bmatrix}
    \\
    & c_4 = \dfrac{1+ \sum_{i=1}^{3}c_i}{\sigma^2}
    \\
    & \bB_4 = - \dfrac{1}{\sigma^2}\begin{bmatrix} \bI & c_3\bI & c_3c_2\bI & c_3c_2c_1\bI \\ \Zero & \Zero & \Zero & \Zero \end{bmatrix}
    \\
    & \bQ_4 = -\dfrac{1}{\sigma^2} \begin{pmatrix} \bI & \Zero & \Zero & \Zero \\
    \Zero & c_3\bI & \Zero & \Zero \\
    \Zero & \Zero & c_3c_2\bI & \Zero \\
    \Zero & \Zero & \Zero & c_3c_2c_1\bI\end{pmatrix}
\end{align}

In the previous quantities, $c_1,c_2,c_3$ represent binary selection variables, accounting for the following possible cases:
\begin{equation}
    \begin{cases}
        c_1=c_2=c_3=1,\,\text{for term }\expec\{s_u(i)s_v(i)s_a(i)s_b(i)\} \\
        c_1=0,c_2=c_3=1,\,\text{for term }\expec\{s_u(i)s_v(i)s_a(i)y_n(i)\} \\
        c_1=c_2=0,c_3=1,\,\text{for terms } \bR_{tt,m} \text{ and }\expec\{s_u(i)s_v(i)y_n^2(i)\} \\
        c_1=c_2=c_3=0,\,\text{for term }\expec\{s_b(i)y_n(i)\}
    \end{cases}\,.
\end{equation}
Note that $h_1,h_2,\ldots,h_8$ are different indexes of~$\widetilde{\by}(i)$ and $\bx_4$ (not necessarily distinct), and the binary variables $\iota_i\in\{0,1\}$, for $i=1,\ldots,4$, allow us to accommodate lower-order cases and be more flexible.

We now have, for $\widetilde{\by}(i)$ jointly Gaussian distributed:
\begin{align}
    & \expec\left\{f(\widetilde{\by}) \triangleq
    (\widetilde{y}_{h_1}-x_{h_2})^{\iota_1} 
    (\widetilde{y}_{h_3}-x_{h_4})^{\iota_2} 
    (\widetilde{y}_{h_5}-x_{h_6})^{\iota_3} 
    (\widetilde{y}_{h_7}-x_{h_8})^{\iota_4}
    \exp\left( -\dfrac{1}{2}c_4 \widetilde{\by}^\top\bB_0\widetilde{\by} - \bx_4^\top \bB_4^\top \widetilde{\by} + \dfrac{1}{2}\bx_4^\top\bQ_4\bx_4 \right) \right\} \nonumber
    \\
    ={} & (2\pi)^{-\dfrac{N+1}{2}}\det\{\bR_{\widetilde{y}}\}^{-\dfrac{1}{2}} \exp\left(\dfrac{1}{2}\bx_4^\top\bQ_4\bx_4\right) \int_{\mathbb{R}}\cdots\int_{\mathbb{R}}
    (\widetilde{y}_{h_1}-x_{h_2})^{\iota_1}
    (\widetilde{y}_{h_3}-x_{h_4})^{\iota_2}
    (\widetilde{y}_{h_5}-x_{h_6})^{\iota_3}
    (\widetilde{y}_{h_7}-x_{h_8})^{\iota_4}
    \nonumber \\ & 
    \qquad \times \exp\left(-\dfrac{1}{2}c_4 \widetilde{\by}^\top\bB_0\widetilde{\by} \right)
    \exp\left(-\frac{1}{2} \widetilde{\by}^\top\bR_{\widetilde{y}}^{-1}\widetilde{\by} \right) \exp\left(-\bx_4^\top\bB_4^\top\widetilde{\by}\right) d\widetilde{y}_1\cdots d\widetilde{y}_{N+1}
    \\
    ={} & (2\pi)^{-\dfrac{N+1}{2}}\det\{\bR_{\widetilde{y}}\}^{-\dfrac{1}{2}} \exp\left(\dfrac{1}{2}\bx_4^\top\bQ_4\bx_4\right) \int_{\mathbb{R}}\cdots\int_{\mathbb{R}} (\widetilde{y}_{h_1}-x_{h_2})^{\iota_1} 
    (\widetilde{y}_{h_3}-x_{h_4})^{\iota_2} 
    (\widetilde{y}_{h_5}-x_{h_6})^{\iota_3}
    (\widetilde{y}_{h_7}-x_{h_8})^{\iota_4}
    \nonumber \\ & 
    \qquad \times \exp\left(-\frac{1}{2} \widetilde{\by}^\top\left[c_4\bB_0+\bR_{\widetilde{y}}^{-1}\right]\widetilde{\by} - \bx_4^\top\bB_4^\top\widetilde{\by} \right) d\widetilde{y}_1\cdots d\widetilde{y}_{N+1}\,.
    \label{eq:orig_integral}
\end{align}

Note that we can write:
\begin{align}
    & \hspace{-1em} -\frac{1}{2} \widetilde{\by}^\top\left[c_4\bB_0+\bR_{\widetilde{y}}^{-1}\right]\widetilde{\by} - \bx_4^\top\bB_4^\top\widetilde{\by}
    \nonumber \\ {}={} &
    -\frac{1}{2} \left( \widetilde{\by}^\top\left[c_4\bB_0+\bR_{\widetilde{y}}^{-1}\right]\widetilde{\by} + \bx_4^\top\bB_4^\top\widetilde{\by} + \widetilde{\by}^\top\bB_4\bx_4\right)
    \nonumber \\ {}={} &
    -\frac{1}{2} \left(\widetilde{\by} + \left[c_4\bB_0+\bR_{\widetilde{y}}^{-1}\right]^{-1}\bB_4\bx_4\right)^\top \left[c_4\bB_0+\bR_{\widetilde{y}}^{-1}\right] \left(\widetilde{\by} + \left[c_4\bB_0+\bR_{\widetilde{y}}^{-1}\right]^{-1}\bB_4\bx_4\right)
    \nonumber \\ &
    +\frac{1}{2} \bx_4^\top\bB_4^\top\left[c_4\bB_0+\bR_{\widetilde{y}}^{-1}\right]^{-1}\bB_4\bx_4\,.
\end{align}
Thus, relation \eqref{eq:orig_integral} becomes:
\begin{align}
    ={} & (2\pi)^{-\dfrac{N+1}{2}}\det\{\bR_{\widetilde{y}}\}^{-\dfrac{1}{2}} \exp\left(\dfrac{1}{2}\bx_4^\top\bQ_4\bx_4\right) \int_{\mathbb{R}}\cdots\int_{\mathbb{R}} (\widetilde{y}_{h_1}-x_{h_2})^{\iota_1} 
    (\widetilde{y}_{h_3}-x_{h_4})^{\iota_2} 
    (\widetilde{y}_{h_5}-x_{h_6})^{\iota_3} 
    (\widetilde{y}_{h_7}-x_{h_8})^{\iota_4}
    \nonumber \\ & 
    \qquad \times \exp\left(-\frac{1}{2} \left(\widetilde{\by} + \left[c_4\bB_0+\bR_{\widetilde{y}}^{-1}\right]^{-1}\bB_4\bx_4\right)^\top \left[c_4\bB_0+\bR_{\widetilde{y}}^{-1}\right] \left(\widetilde{\by} + \left[c_4\bB_0+\bR_{\widetilde{y}}^{-1}\right]^{-1}\bB_4\bx_4\right) \right) 
    \nonumber \\ & \qquad 
    \times \exp\left(\frac{1}{2} \bx_4^\top\bB_4^\top\left[c_4\bB_0+\bR_{\widetilde{y}}^{-1}\right]^{-1}\bB_4\bx_4\right) d\widetilde{y}_1\cdots d\widetilde{y}_{N+1}
    \nonumber \\ 
    ={} & (2\pi)^{-\dfrac{N+1}{2}}\det\{\bR_{\widetilde{y}}\}^{-\dfrac{1}{2}} \exp\left(\dfrac{1}{2}\bx_4^\top\bQ_4\bx_4\right) \det\left\{\left[c_4\bB_0+\bR_{\widetilde{y}}^{-1}\right]^{-1}\right\}^{\dfrac{1}{2}} 
    \det\left\{\left[c_4\bB_0+\bR_{\widetilde{y}}^{-1}\right]^{-1}\right\}^{-\dfrac{1}{2}}
    \nonumber \\ & 
    \qquad \times \exp\left(\frac{1}{2} \bx_4^\top\bB_4^\top\left[c_4\bB_0+\bR_{\widetilde{y}}^{-1}\right]^{-1}\bB_4\bx_4\right)
    \int_{\mathbb{R}}\cdots\int_{\mathbb{R}} 
    (\widetilde{y}_{h_1}-x_{h_2})^{\iota_1} 
    (\widetilde{y}_{h_3}-x_{h_4})^{\iota_2} 
    (\widetilde{y}_{h_5}-x_{h_6})^{\iota_3} 
    (\widetilde{y}_{h_7}-x_{h_8})^{\iota_4}
    \nonumber \\ & 
    \qquad \times \exp\left(-\frac{1}{2} \left(\widetilde{\by} + \left[c_4\bB_0+\bR_{\widetilde{y}}^{-1}\right]^{-1}\bB_4\bx_4\right)^\top \left[c_4\bB_0+\bR_{\widetilde{y}}^{-1}\right] \left(\widetilde{\by} + \left[c_4\bB_0+\bR_{\widetilde{y}}^{-1}\right]^{-1}\bB_4\bx_4\right) \right) d\widetilde{y}_1\cdots d\widetilde{y}_{N+1}
    \nonumber \\ 
    ={} & \det\{\bR_{\widetilde{y}}\}^{-\dfrac{1}{2}}
    \det\left\{\left[c_4\bB_0+\bR_{\widetilde{y}}^{-1}\right]^{-1}\right\}^{\dfrac{1}{2}} \exp\left(\frac{1}{2} \bx_4^\top\left(\bQ_4+\bB_4^\top\left[c_4\bB_0+\bR_{\widetilde{y}}^{-1}\right]^{-1}\bB_4\right)\bx_4\right)
    \nonumber \\ & 
    \qquad \times \expec_{p'(\widetilde{\by})} \Big\{  
    (\widetilde{y}_{h_1}-x_{h_2})^{\iota_1} 
    (\widetilde{y}_{h_3}-x_{h_4})^{\iota_2} 
    (\widetilde{y}_{h_5}-x_{h_6})^{\iota_3} 
    (\widetilde{y}_{h_7}-x_{h_8})^{\iota_4}\Big\} 
    \label{eq:expec_suppl} \\ 
    {}={} & \nu(\{h_i\}_{i=1}^8)\,,
\end{align}
with the normal distribution:
\begin{align}
    p'(\widetilde{\by}) = \mathcal{N}\left( \bmu =  -\left[c_4\bB_0+\bR_{\widetilde{y}}^{-1}\right]^{-1}\bB_4\bx_4, \, \bSig = \left[c_4\bB_0+\bR_{\widetilde{y}}^{-1}\right]^{-1} \right)\,.
\end{align}

Let us focus on the expectation in \eqref{eq:expec_suppl}. We have, for $\iota_i=1, i = \{1,2,3,4\}$:
\begin{align}
    & \hspace{-1em} \expec_{p'(\by)} \Big\{  (\widetilde{y}_{h_1}-x_{h_2})^{\iota_1} (\widetilde{y}_{h_3}-x_{h_4})^{\iota_2} (\widetilde{y}_{h_5}-x_{h_6})^{\iota_3} (\widetilde{y}_{h_7}-x_{h_8})^{\iota_4}\Big\} \nonumber \\ = {} & \expec\big\{
    \widetilde{y}_{h_1}\widetilde{y}_{h_3}\widetilde{y}_{h_5}\widetilde{y}_{h_7}\big\}
    -x_{h_8}\expec\big\{\widetilde{y}_{h_1}\widetilde{y}_{h_3}\widetilde{y}_{h_5}\big\}
    -x_{h_6}\expec\big\{\widetilde{y}_{h_1}\widetilde{y}_{h_3}\widetilde{y}_{h_7}\big\}
    +x_{h_6}x_{h_8}\expec\big\{\widetilde{y}_{h_1}\widetilde{y}_{h_3}\big\}
    \nonumber \\  & %
    -x_{h_4}\expec\big\{\widetilde{y}_{h_1}\widetilde{y}_{h_5}\widetilde{y}_{h_7}\big\}
    +x_{h_4}x_{h_8}\expec\big\{\widetilde{y}_{h_1}\widetilde{y}_{h_5}\big\}
    +x_{h_4}x_{h_6}\expec\big\{\widetilde{y}_{h_1}\widetilde{y}_{h_7}\big\}
    -x_{h_4}x_{h_6}x_{h_8}\expec\big\{\widetilde{y}_{h_1}\big\}
    \nonumber \\ & %
    -x_{h_2}\expec\big\{\widetilde{y}_{h_3}\widetilde{y}_{h_5}\widetilde{y}_{h_7}\big\}
    +x_{h_2}x_{h_8}\expec\big\{\widetilde{y}_{h_3}\widetilde{y}_{h_5}\big\}
    +x_{h_2}x_{h_6}\expec\big\{\widetilde{y}_{h_3}\widetilde{y}_{h_7}\big\}
    -x_{h_2}x_{h_6}x_{h_8}\expec\big\{\widetilde{y}_{h_3}\big\}
    \nonumber \\ & %
    +x_{h_2}x_{h_4}\expec\big\{\widetilde{y}_{h_5}\widetilde{y}_{h_7}\big\}
    -x_{h_2}x_{h_4}x_{h_8}\expec\big\{\widetilde{y}_{h_5}\big\}
    -x_{h_2}x_{h_4}x_{h_6}\expec\big\{\widetilde{y}_{h_7}\big\}
    +x_{h_2}x_{h_4}x_{h_6}x_{h_8}\,.
\end{align}
Since $h_i, i\in \{1,3,5,7\}$ represent solely generic placeholders for any actual index, it suffices to compute only a few expectations from the previous relation, which can then be used for any other combination.
We have:
\begin{align}
    \expec_{p'(\widetilde{\by})} \Big\{ \widetilde{y}_{h_1} \widetilde{y}_{h_3} \widetilde{y}_{h_5} \widetilde{y}_{h_7} \Big\}
    {}={} & \mu_{h_1}\mu_{h_3}\mu_{h_5}\mu_{h_7} + 
    \Sigma_{h_1,h_3}\Sigma_{h_5,h_7}
    + \Sigma_{h_1,h_5}\Sigma_{h_3,h_7}
    + \Sigma_{h_1,h_7}\Sigma_{h_3,h_5}
    \nonumber \\ & + 
    \mu_{h_1}\mu_{h_3}\Sigma_{h_5,h_7} + \mu_{h_1}\mu_{h_5}\Sigma_{h_3,h_7} + \mu_{h_1}\mu_{h_7}\Sigma_{h_3,h_5}
    \nonumber \\ & +
    \mu_{h_3}\mu_{h_5}\Sigma_{h_1,h_7} + \mu_{h_3}\mu_{h_7}\Sigma_{h_1,h_5} + \mu_{h_5}\mu_{h_7}\Sigma_{h_1,h_3}\,,
\\
    \expec_{p'(\by)} \Big\{ \widetilde{y}_{h_1} \widetilde{y}_{h_3} \widetilde{y}_{h_5} \Big\}
    {}={} & \mu_{h_1}\mu_{h_3}\mu_{h_5} + 
    \mu_{h_5}\Sigma_{h_1,h_3} + \mu_{h_3}\Sigma_{h_1,h_5} + \mu_{h_1}\Sigma_{h_3,h_5}\,,
\\ 
    \expec_{p'(\by)} \Big\{ \widetilde{y}_{h_1} \widetilde{y}_{h_3} \Big\}
    {}={} & \mu_{h_1}\mu_{h_3} + 
    \Sigma_{h_1,h_3}\,,
\\
    \expec_{p'(\by)}\Big\{ \widetilde{y}_{h_1}\Big\}
    {}={} & \mu_{h_1}\,.
\end{align}

\section{Cases corresponding to \texorpdfstring{$  \bR_{tt,m}$ ($c_1=c_2=0,c_3=1$)}{Rtt,m = E\{Tm\} (c1=c2=0,c3=1)}}
\label{appendix:Rttm}

Since $\bt_m(p) = \begin{bmatrix} \bl_m(p) \\ \bzeta_m(p) \end{bmatrix}$, the matrix $\bR_{tt,m}$ can be written as:
\begin{align}
    \bR_{tt,m} & = \expec\{\bt_m(i)\bt_m^\top\!(i)\}
    \nonumber\\ &
    = \expec\left\{ \begin{bmatrix} \bl_m(i) \\ \bzeta_m(i) \end{bmatrix}
    \begin{bmatrix} \bl_m^\top\!(i) & \bzeta_m^\top\!(i) \end{bmatrix}\right\}
    \nonumber\\ &
    = \begin{bmatrix} \bR_{\ell\ell,m} &  \bR_{\zeta\ell,m}^\top \\ \bR_{\zeta\ell,m} &  \bR_{\zeta\zeta,m}
    \end{bmatrix}\,,
\end{align}
where in the last step we used the fact that the signals are identically distributed (for different~$i$).

For ease of comprehension, we recall the definitions of $\bz(i)$, $\bzeta(i)$ and $\bl_m(i)$:
\begin{align*}
\begin{array}{ll}
   \bz(i) = \big[\bz_1^\top\!(i),\ldots, \bz_{N}^\top\!(i)\big]^\top\,, \hspace{1cm}
   &
   [\bz_m(i)]_{q} = \dfrac{\partial \kappa(\by(i),\by(\omega_q))}{\partial y_{m}(\omega_q)}\Bigg\vert_{q=1,\ldots,\vert \cD_n\vert}\,, \\
    \bzeta(i) = \big[\bzeta_1^\top\!(i),\ldots,\bzeta_N^\top\!(i)\big]^\top\,,
    &
    [\bzeta_m(i)]_{q} = \dfrac{\partial \kappa(\by(i),\by(\omega_q))}{\partial y_{m}(i)}\Bigg\vert_{q=1,\ldots,\vert \cD_n \vert}\,, \\
    \bl_m(i) = \big[\bl_{1,m}^\top\!(i), \ldots, \bl_{N,m}^\top\!(i)\big]^\top\,,
    &
     [\bl_{m_1,m_2}(i)]_q = \dfrac{\partial^2 \kappa (\by(i),\by(\omega_q))}{\partial y_{m_1}(\omega_q) \partial y_{m_2}(\omega_q)}\Bigg\vert_{q=1,\ldots,\vert \cD_n \vert}\,.
\end{array}
\end{align*}
Their explicit forms for the particular case of the Gaussian kernel are:
\begin{align*}
    [\bz_m(i)]_q &= -[\bzeta_m(i)]_q 
    \\
    & = \exp \left( -\dfrac{1}{2\sigma^2} \Vert \by(i) - \by(\omega_q) \Vert^2 \right)\left(\dfrac{[\by(i)]_m - [\by(\omega_q)]_m}{\sigma^2} \right)\,,
\\
    [\bl_{m_1,m_2}]_q & = \begin{cases}
    -\exp \left( -\dfrac{1}{2\sigma^2} \Vert \by(i) - \by(\omega_q) \Vert^2 \right) \dfrac{[\by(i)]_{m_1} - [\by(\omega_q)]_{m_1}}{\sigma^2} \dfrac{[\by(i)]_{m_2} - [\by(\omega_q)]_{m_2}}{\sigma^2} & m_1 \neq m_2  \\
    -\exp \left( -\dfrac{1}{2\sigma^2} \Vert \by(i) - \by(\omega_q) \Vert^2 \right) \left( \dfrac{\big([\by(i)]_{m_1} - [\by(\omega_q)]_{m_1} \big)^2}{\sigma^4} - \dfrac{1}{\sigma^2} \right) & m_1 = m_2
\end{cases}\,.
\end{align*}

\bigskip
\subsection{Computing \texorpdfstring{$\bR_{\ell\ell,m}$}{Rllm}}

\begin{align}
    \big[\bR_{\ell\ell,m} & \big]_{(a-1)|\cD_n|+p,\,(b-1)|\cD_n|+q}
     \nonumber \\[0.2cm] & 
     = \big[\expec\{\bl_m(i)\bl_m^\top\!(i)\}\big]_{(a-1)|\cD_n|+p,\,(b-1)|\cD_n|+q} 
    \nonumber \\[0.2cm]  & 
     = \expec\big\{[\bl_{a,m}]_p [\bl_{b,m}]_q\big\} \nonumber\\[0.2cm]
    & =
    \begin{cases}
        \begin{array}{c}\expec\bigg\{ 
        \dfrac{(y_a(i) - y_a(\omega_p))}{\sigma^2} 
        \dfrac{(y_m(i) - y_m(\omega_p))}{\sigma^2} 
        \dfrac{(y_b(i) - y_b(\omega_q))}{\sigma^2} 
        \dfrac{(y_m(i) - y_m(\omega_q))}{\sigma^2} 
        \\
        \times \exp \left( -\dfrac{1}{2\sigma^2} \Vert \by(i) - \by(\omega_p) \Vert^2  -\dfrac{1}{2\sigma^2} \Vert \by(i) - \by(\omega_q) \Vert^2 \right) 
        \bigg\}
        \end{array}, & a\neq m, \, b\neq m
        \\
        \begin{array}{c}\expec\bigg\{ 
        \dfrac{(y_a(i) - y_a(\omega_p))}{\sigma^2} 
        \dfrac{(y_m(i) - y_m(\omega_p))}{\sigma^2} 
        \left( \dfrac{\big(y_b(i) - y_b(\omega_q) \big)^2}{\sigma^4} - \dfrac{1}{\sigma^2} \right)
        \\
        \times \exp \left( -\dfrac{1}{2\sigma^2} \Vert \by(i) - \by(\omega_p) \Vert^2  -\dfrac{1}{2\sigma^2} \Vert \by(i) - \by(\omega_q) \Vert^2 \right) 
        \bigg\}
        \end{array}, & a\neq m, \, b= m
        \\
        \begin{array}{c}\expec\bigg\{ 
        \dfrac{(y_b(i) - y_b(\omega_q))}{\sigma^2} 
        \dfrac{(y_m(i) - y_m(\omega_q))}{\sigma^2} 
        \left( \dfrac{\big(y_a(i) - y_a(\omega_p) \big)^2}{\sigma^4} - \dfrac{1}{\sigma^2} \right)
        \\
        \times \exp \left( -\dfrac{1}{2\sigma^2} \Vert \by(i) - \by(\omega_p) \Vert^2  -\dfrac{1}{2\sigma^2} \Vert \by(i) - \by(\omega_q) \Vert^2 \right) 
        \bigg\}
        \end{array}, & a= m, \, b\neq m
        \\
        \begin{array}{c}\expec\bigg\{ 
        \left( \dfrac{\big(y_a(i) - y_a(\omega_p) \big)^2}{\sigma^4} - \dfrac{1}{\sigma^2} \right) 
        \left( \dfrac{\big(y_b(i) - y_b(\omega_q) \big)^2}{\sigma^4} - \dfrac{1}{\sigma^2} \right)
        \\
        \times \exp \left( -\dfrac{1}{2\sigma^2} \Vert \by(i) - \by(\omega_p) \Vert^2  -\dfrac{1}{2\sigma^2} \Vert \by(i) - \by(\omega_q) \Vert^2 \right) 
        \bigg\}
        \end{array}, & a= m, \, b= m
    \end{cases}\,,
    \label{eq:expec_Rllm_cases}
\end{align}
which is equivalent to:
\begin{align}
\big[\bR_{\ell\ell,m} &\big]_{(a-1)|\cD_n|+p,\,(b-1)|\cD_n|+q} 
\nonumber\\
&  = \begin{cases}
    \dfrac{1}{\sigma^8} \nu(\{h_i\}_{i=1}^8),\, \text{with }
    \begin{cases} h_1=a \\ h_2=a \\ h_3=m \\ h_4=m \\ h_5=b \\ h_6=N+b \\ h_7=m \\ h_8=N+m
    \end{cases}, & \text{when } a\neq m,\, b\neq m
    \\
    \begin{array}{l}
    \dfrac{1}{\sigma^8} \nu(\{h_i\}_{i=1}^8) \\
    -\dfrac{1}{\sigma^2}\big[\bR_{zz}\big]_{(a-1)|\cD_n|+p,(m-1)|\cD_n|+p}
    \end{array}
    ,\, \text{with } \begin{cases} h_1=a \\ h_2=a \\ h_3=b \\ h_4=b \\ h_5=b \\ h_6=N+b \\ h_7=b \\ h_8=N+b
    \end{cases}, & \text{when } a\neq m,\, b=m
    \\
    \begin{array}{l}
    \dfrac{1}{\sigma^8} \nu(\{h_i\}_{i=1}^8) \\
    -\dfrac{1}{\sigma^2}\big[\bR_{zz}\big]_{(b-1)|\cD_n|+q,(m-1)|\cD_n|+q}
    \end{array}
    ,\, \text{with } \begin{cases} h_1=b \\ h_2=b \\ h_3=a \\ h_4=a \\ h_5=a \\ h_6=N+a \\ h_7=a \\ h_8=N+a
    \end{cases}, & \text{when } a=m,\, b\neq m
    \\
    \begin{array}{l}
    \dfrac{1}{\sigma^8} \nu(\{h_i\}_{i=1}^8) \\
    -\dfrac{1}{\sigma^2}\big[\bR_{zz}\big]_{(a-1)|\cD_n|+p,(a-1)|\cD_n|+p} \\
    -\dfrac{1}{\sigma^2}\big[\bR_{zz}\big]_{(b-1)|\cD_n|+q,(b-1)|\cD_n|+q}\\
    +\dfrac{1}{\sigma^4}\big[\bR_{kk}\big]_{p,q}
    \end{array}
    ,\, \text{with } \begin{cases} h_1=a \\ h_2=a \\ h_3=a \\ h_4=a \\ h_5=b \\ h_6=N+b \\ h_7=b \\ h_8=N+b
    \end{cases}, & \text{when } a=m,\, b=m
\end{cases}\,.
\end{align}

\subsection{Computing \texorpdfstring{$\bR_{\zeta\ell,m}$}{Rzlm}}
\begin{align}
    \big[\bR_{\zeta\ell,m} & \big]_{p,\,(b-1)|\cD_n|+q} 
         \nonumber \\[0.2cm] & 
     = \big[\expec\{\bzeta_m(i)\bl_m^\top\!(i)\}\big]_{p,\,(b-1)|\cD_n|+q} 
          \nonumber \\[0.2cm] & 
     = \expec\big\{[\bzeta_m]_p [\bl_{b,m}]_q\big\}
    \nonumber\\[0.2cm]
    & = 
    \begin{cases}
        \begin{array}{l} - \expec\bigg\{ 
        \dfrac{(y_m(i) - y_m(\omega_p))}{\sigma^2} 
        \dfrac{(y_b(i) - y_b(\omega_q))}{\sigma^2} 
        \dfrac{(y_m(i) - y_m(\omega_q))}{\sigma^2} 
        \\
        \times \exp \left( -\dfrac{1}{2\sigma^2} \Vert \by(i) - \by(\omega_p) \Vert^2  -\dfrac{1}{2\sigma^2} \Vert \by(i) - \by(\omega_q) \Vert^2 \right) 
        \bigg\}
        \end{array}, & b\neq m
        \\
        \begin{array}{l} - \expec\bigg\{ 
        \dfrac{(y_b(i) - y_b(\omega_p))}{\sigma^2} 
        \left( \dfrac{\big(y_b(i) - y_b(\omega_q) \big)^2}{\sigma^4} - \dfrac{1}{\sigma^2} \right)
        \\
        \times \exp \left( -\dfrac{1}{2\sigma^2} \Vert \by(i) - \by(\omega_p) \Vert^2  -\dfrac{1}{2\sigma^2} \Vert \by(i) - \by(\omega_q) \Vert^2 \right) 
        \bigg\}
        \end{array}, & b= m
    \end{cases}\,.
    \label{eq:expec_Rzlm_cases}
\end{align}
which becomes:
\begin{gather}
\big[\bR_{\zeta\ell,m}\big]_{p,\,(b-1)|\cD_n|+q} =
\begin{cases}
    - \dfrac{1}{\sigma^6} \nu(\{h_i\}_{i=1}^8),\, \text{with }
    \begin{cases} 
    \begin{rcases}h_1=\bullet \\ h_2=\bullet\end{rcases} \iota_1=0 \\ h_3=m \\ h_4=m \\ h_5=b \\ h_6=N+b \\ h_7=m \\ h_8=N+m
    \end{cases}, & \text{when } b\neq m
    \\
    \begin{array}{l}
    - \dfrac{1}{\sigma^6} \nu(\{h_i\}_{i=1}^8) \\
    + \dfrac{1}{\sigma^2}\big[\bR_{kz}\big]_{q,(b-1)|\cD_n|+p}
    \end{array}
    ,\, \text{with } \begin{cases} \begin{rcases} h_1=\bullet \\ h_2=\bullet\end{rcases} \iota_1=0 \\ h_3=b \\ h_4=b \\ h_5=b \\ h_6=N+b \\ h_7=b \\ h_8=N+b
    \end{cases}, & \text{when } b=m
\end{cases}\,,
\end{gather}
where the bullet $\bullet$ means that the index is irrelevant, due to its corresponding term not being present in the generic form \eqref{eq:generic_form}.

\subsection{Computing \texorpdfstring{$\bR_{ss}$ (including $\bR_{zz,m}$, $\bR_{zz,m_1,m_2}$, $\bR_{kz,m}$ and $\bR_{kk,m}$)}{Rss (including Rzzm, Rzzm1m2, Rkzm and Rkkm)}} \label{subsec:Rss}

Since $\bs(i) = \begin{bmatrix}\bz(i)\\ \bk(i)\end{bmatrix}$, all the matrices detailed in this section can be used in computing $\bR_{ss}$.

\begin{align}
    \big[\bR_{\zeta\zeta,m}\big]_{p,\,q} &
     = \big[\expec\{\bzeta_m(i)\bzeta_m^\top\!(i)\}\big]_{p,\,q} 
     \nonumber\\[0.2cm] &
     = \expec\big\{[\bzeta_m]_p [\bzeta_m]_q\big\}
    \nonumber\\[0.2cm]
    & = 
        \begin{array}{c}\expec\bigg\{ 
        \dfrac{(y_m(i) - y_m(\omega_p))}{\sigma^2} 
        \dfrac{(y_m(i) - y_m(\omega_q))}{\sigma^2} 
        \exp \left( -\dfrac{1}{2\sigma^2} \Vert \by(i) - \by(\omega_p) \Vert^2  -\dfrac{1}{2\sigma^2} \Vert \by(i) - \by(\omega_q) \Vert^2 \right) 
        \bigg\}
        \end{array}\,.
    \label{eq:expec_Rzzm_cases}
\\ &
    = \dfrac{1}{\sigma^4} \nu(\{h_i\}_{i=1}^8),\, \text{with }
    \begin{cases} 
    \begin{rcases}h_1=\bullet \\ h_2=\bullet\end{rcases} \iota_1=0 \\ h_3=m \\ h_4=m \\ \begin{rcases}h_5=\bullet \\ h_6=\bullet\end{rcases} \iota_3=0 \\ h_7=m \\ h_8=N+m
    \end{cases}\,.
\end{align}

Generalizing to the case $m_1\neq m_2$, we straightforwardly have:
\begin{align}
    \big[\bR_{\zeta\zeta,m_1,m_2}\big]_{p,\,q} &
     = \big[\expec\{\bzeta_{m_1}(i)\bzeta_{m_2}^\top\!(i)\}\big]_{p,\,q} 
     \nonumber\\[0.2cm] &
     = \expec\big\{[\bzeta_{m_1}]_p [\bzeta_{m_2}]_q\big\}
    \nonumber\\[0.2cm]
    & = 
        \begin{array}{c}\expec\bigg\{ 
        \dfrac{(y_{m_1}(i) - y_{m_1}(\omega_p))}{\sigma^2} 
        \dfrac{(y_{m_2}(i) - y_{m_2}(\omega_q))}{\sigma^2} 
        \exp \left( -\dfrac{1}{2\sigma^2} \Vert \by(i) - \by(\omega_p) \Vert^2  -\dfrac{1}{2\sigma^2} \Vert \by(i) - \by(\omega_q) \Vert^2 \right) 
        \bigg\}
        \end{array}\,.
    \label{eq:expec_Rzzm_1m_2_cases}
\\ &
    = \dfrac{1}{\sigma^4} \nu(\{h_i\}_{i=1}^8),\, \text{with }
    \begin{cases} 
    \begin{rcases}h_1=\bullet \\ h_2=\bullet\end{rcases} \iota_1=0 \\ h_3=m_1 \\ h_4=m_1 \\ \begin{rcases}h_5=\bullet \\ h_6=\bullet\end{rcases} \iota_3=0 \\ h_7=m_2 \\ h_8=N+m_2
    \end{cases}\,.
\end{align}

Particularizing to $\bR_{kz,m}$ and using the fact that for the Gaussian kernel $[\bz_m]_q = - [\bzeta_m]_q$, we have:
\begin{align}
    \big[\bR_{kz,m}\big]_{p,\,q} &
     = - \big[\expec\{\bk(i)\bzeta_m^\top\!(i)\}\big]_{p,\,q}
     \nonumber\\[0.2cm] &
     = \expec\big\{[\bk]_p [\bz_m]_q\big\}
    \nonumber\\[0.2cm]
    & = 
        \expec\bigg\{
        \dfrac{(y_{m_2}(i) - y_{m_2}(\omega_q))}{\sigma^2} 
        \exp \left( -\dfrac{1}{2\sigma^2} \Vert \by(i) - \by(\omega_p) \Vert^2  -\dfrac{1}{2\sigma^2} \Vert \by(i) - \by(\omega_q) \Vert^2 \right) 
        \bigg\}
        \,.
    \label{eq:expec_Rkzm_cases}
\\ &
    = \dfrac{1}{\sigma^2} \nu(\{h_i\}_{i=1}^8),\, \text{with }
    \begin{cases} 
    \begin{rcases}h_1=\bullet \\ h_2=\bullet\end{rcases} \iota_1=0 \\ \begin{rcases}h_3=\bullet \\ h_4=\bullet \end{rcases} \iota_2=0\\ \begin{rcases}h_5=\bullet \\ h_6=\bullet\end{rcases} \iota_3=0 \\ h_7=m \\ h_8=N+m
    \end{cases}\,.
\end{align}

Particularizing to $\bR_{kk}$, we have:
\begin{align}
    \big[\bR_{kk}\big]_{p,\,q} &
     = \big[\expec\{\bk(i)\bk^\top\!(i)\}\big]_{p,\,q}
     \nonumber\\[0.2cm] &
     = \expec\big\{[\bk]_p [\bk]_q\big\}
    \nonumber\\[0.2cm]
    & = 
        \begin{array}{c}\expec\bigg\{ 
        \exp \left( -\dfrac{1}{2\sigma^2} \Vert \by(i) - \by(\omega_p) \Vert^2  -\dfrac{1}{2\sigma^2} \Vert \by(i) - \by(\omega_q) \Vert^2 \right) 
        \bigg\}
        \end{array}\,.
    \label{eq:expec_Rkk_cases}
\\ & 
    = \nu(\{h_i\}_{i=1}^8),\, \text{with }
    \begin{cases} 
    \begin{rcases}h_1=\bullet \\ h_2=\bullet\end{rcases} \iota_1=0 \\ \begin{rcases}h_3=\bullet \\ h_4=\bullet \end{rcases} \iota_2=0\\ \begin{rcases}h_5=\bullet \\ h_6=\bullet\end{rcases} \iota_3=0 \\ \begin{rcases}h_7=\bullet \\ h_8=\bullet\end{rcases} \iota_4=0
    \end{cases}\,.
\end{align}


\section{Cases corresponding to \texorpdfstring{$\expec\{s_u(i) s_{v}(i) s_a(i) s_b(i)\}$ ($c_1=c_2=c_3=1$)}{E\{su(i)sv(i)sa(i)sb(i)\} (c1=c2=c3=1)}} \label{subsec:suxyv}

Recall that $\bs(i) = \begin{bmatrix}\bz(i)\\ \bk(i)\end{bmatrix}$. We have the following terms: \

\paragraph{Term $\expec\{z_{u}(i)z_{a}(i)z_{b}(i)z_{v}(i)\}$} $u,v = 1,\ldots,N|\cD_n|$, $a,b = 1,\ldots,N|\cD_n|$
\begin{align}
    \expec\{z_{u}(i)& z_{a}(i)z_{a}(i)z_{v}(i)\} \nonumber\\ {}={} &
    \expec\bigg\{ 
        \dfrac{(y_{m_1}(i) - y_{m_1}(\omega_p))}{\sigma^2} 
        \dfrac{(y_{m_2}(i) - y_{m_2}(\omega_q))}{\sigma^2} 
        \dfrac{(y_{m_3}(i) - y_{m_3}(\omega_r))}{\sigma^2} 
        \dfrac{(y_{m_4}(i) - y_{m_4}(\omega_s))}{\sigma^2} 
        \nonumber\\ &
        \times \exp \left( -\dfrac{1}{2\sigma^2} \left(\Vert \by(i) - \by(\omega_p) \Vert^2 + \Vert \by(i) - \by(\omega_q) \Vert^2 + \Vert \by(i) - \by(\omega_r) \Vert^2 + \Vert \by(i) - \by(\omega_s) \Vert^2\right)\right)
        \bigg\}
        \label{eq:zzzy1}\\ {}={} &
        \dfrac{1}{\sigma^8} \nu(\{h_i\}_{i=1}^8),\, \text{with }
    \begin{cases} 
    h_1=h_2=m_1=\left\lceil \frac{u}{|\cD_n|} \right\rceil \\ h_3=h_4-N=m_2=\left\lceil \frac{a}{|\cD_n|} \right\rceil \\ h_5=h_6-2N=m_3=\left\lceil \frac{b}{|\cD_n|} \right\rceil \\ h_7=h_8-3N=m_4=\left\lceil \frac{v}{|\cD_n|} \right\rceil
    \end{cases}\,,
    \begin{cases} 
    \omega_p=\modop(u-1,|\cD_n|)+1 \\ \omega_q=\modop(a-1,|\cD_n|)+1 \\ \omega_r=\modop(b-1,|\cD_n|)+1 \\ \omega_s=\modop(v-1,|\cD_n|)+1
    \end{cases}\,.
    \label{eq:zzzy2}
\end{align}

\paragraph{Term $\expec\{k_{u}(i)z_{a}(i)z_{b}(i)z_{v}(i)\}$} $u = 1,\ldots, |\cD_n|, v = 1,\ldots, N|\cD_n|$, $a,b = 1,\ldots,N|\cD_n|$
\begin{align}
    \expec\{k_{u}(i)z_{a}(i)z_{b}(i)z_{v}(i)\}
    \nonumber
    {}={} &
    \expec\bigg\{ 
        \dfrac{(y_{m_2}(i) - y_{m_2}(\omega_q))}{\sigma^2} 
        \dfrac{(y_{m_3}(i) - y_{m_3}(\omega_r))}{\sigma^2} 
        \dfrac{(y_{m_4}(i) - y_{m_4}(\omega_s))}{\sigma^2} 
        \exp \bigg( -\dfrac{1}{2\sigma^2} \Big(\Vert \by(i) - \by(\omega_p) \Vert^2 
        \\ &
        + \Vert \by(i) - \by(\omega_q) \Vert^2 + \Vert \by(i) - \by(\omega_r) \Vert^2 + \Vert \by(i) - \by(\omega_s) \Vert^2\Big)\bigg)
        \bigg\}
        \label{eq:zkzy1}\\ {}={} &
        \dfrac{1}{\sigma^6} \nu(\{h_i\}_{i=1}^8),\, \text{with }
    \begin{cases} 
    h_1=h_2=m_1=\bullet \\ h_3=h_4-N=m_2=\left\lceil \frac{a}{|\cD_n|} \right\rceil \\ h_5=h_6-2N=m_3=\left\lceil \frac{b}{|\cD_n|} \right\rceil \\ h_7=h_8-3N=m_4=\left\lceil \frac{v}{|\cD_n|} \right\rceil 
    \end{cases}\,,
    \begin{cases} 
    \omega_p=u \\ \omega_q=\modop(a-1,|\cD_n|)+1 \\ \omega_r=\modop(b-1,|\cD_n|)+1 \\ \omega_s=\modop(v-1,|\cD_n|)+1
    \end{cases}\,.
    \label{eq:zkzy2}
\end{align}

\paragraph{Term $\expec\{k_{u}(i)k_{a}(i)z_{b}(i)z_{v}(i)\}$} $u = 1,\ldots, |\cD_n|, v = 1,\ldots, N|\cD_n|$, $a = 1,\ldots,|\cD_n|, b = 1,\ldots,N|\cD_n|$
\begin{align}
    \expec\{k_{u}(i)k_{a}(i)z_{b}(i)z_{v}(i)\}
    \nonumber
     {}={} &
    \expec\bigg\{ 
        \dfrac{(y_{m_3}(i) - y_{m_3}(\omega_r))}{\sigma^2} 
        \dfrac{(y_{m_4}(i) - y_{m_4}(\omega_s))}{\sigma^2} 
        \exp \bigg( -\dfrac{1}{2\sigma^2} \Big(\Vert \by(i) - \by(\omega_p) \Vert^2 
        \\ &
        + \Vert \by(i) - \by(\omega_q) \Vert^2 + \Vert \by(i) - \by(\omega_r) \Vert^2 + \Vert \by(i) - \by(\omega_s) \Vert^2\Big)\bigg)
        \bigg\}
        \label{eq:zkky1}\\ {}={} &
        \dfrac{1}{\sigma^4} \nu(\{h_i\}_{i=1}^8),\, \text{with }
    \begin{cases} 
    h_1=h_2=m_1=\bullet \\ h_3=h_4-N=m_2=\bullet \\ h_5=h_6-2N=m_3=\left\lceil \frac{b}{|\cD_n|} \right\rceil \\ h_7=h_8-3N=m_4=\left\lceil \frac{v}{|\cD_n|} \right\rceil
    \end{cases}\,,
    \begin{cases} 
    \omega_p=u \\ \omega_q=a \\ \omega_r=\modop(b-1,|\cD_n|)+1 \\ \omega_s=\modop(v-1,|\cD_n|)+1
    \end{cases}\,.
    \label{eq:zkky2}
\end{align}

\paragraph{Term $\expec\{k_{u}(i)k_{a}(i)k_{b}(i)z_{v}(i)\}$} $u = 1,\ldots, |\cD_n|, v = 1,\ldots, N|\cD_n|$, $a,b = 1,\ldots,|\cD_n|$
\begin{align}
    \expec\{k_{u}(i)k_{a}(i)k_{b}(i)z_{v}(i)\} 
    \nonumber 
    {}={} &
    \expec\bigg\{ 
        \dfrac{(y_{m_4}(i) - y_{m_4}(\omega_s))}{\sigma^2} 
        \exp \bigg( -\dfrac{1}{2\sigma^2} \Big(\Vert \by(i) - \by(\omega_p) \Vert^2
        \\ &
         + \Vert \by(i) - \by(\omega_q) \Vert^2 + \Vert \by(i) - \by(\omega_r) \Vert^2 + \Vert \by(i) - \by(\omega_s) \Vert^2\Big)\bigg)
        \bigg\}
        \\ {}={} &
    \dfrac{1}{\sigma^2} \nu(\{h_i\}_{i=1}^8),\, \text{with }
    \begin{cases} 
    h_1=h_2=m_1=\bullet \\ h_3=h_4-N=m_2=\bullet \\ h_5=h_6-2N=m_3=\bullet \\ h_7=h_8-3N=m_4=\left\lceil \frac{v}{|\cD_n|} \right\rceil 
    \end{cases}\,,
    \begin{cases} 
    \omega_p= u \\ \omega_q=a \\ \omega_r=b \\ \omega_s=\modop(v-1,|\cD_n|)+1
    \end{cases}\,.
\end{align}

\paragraph{Term $\expec\{k_{u}(i)k_{a}(i)k_{a}(i)k_{v}(i)\}$} $u,v = 1,\ldots, |\cD_n|$, $a,b = 1,\ldots,|\cD_n|$
\begin{align}
    \expec\{&k_{u}(i)k_{a}(i)k_{b}(i)k_{v}(i)\}\nonumber\\ & = \exp \left( -\dfrac{1}{2\sigma^2} \left(\Vert \by(i) - \by(\omega_p) \Vert^2 + \Vert \by(i) - \by(\omega_q) \Vert^2 + \Vert \by(i) - \by(\omega_r) \Vert^2 + \Vert \by(i) - \by(\omega_s) \Vert^2\right)\right) \bigg\} \label{eq:kkky1}\\ & = \nu(\{h_i\}_{i=1}^8),\, \text{with }
    \begin{cases} 
    h_1=h_2=m_1=\bullet \\ h_3=h_4-N=m_2=\bullet \\ h_5=h_6-2N=m_3=\bullet \\ h_7=h_8-3N=m_4=\bullet
    \end{cases}\,,
    \begin{cases} 
    \omega_p=u \\ \omega_q=a \\ \omega_r=b \\ \omega_s=v
    \end{cases}\,.
    \label{eq:kkky2}
\end{align}


\section{Cases corresponding to \texorpdfstring{$\expec\{s_u(i) s_{a}(i) s_b(i) y_n(i)\}$ ($c_1=0,c_2=c_3=1$)}{E\{su(i)sa(i)sb(i)yn(i)\} (c1=0,c2=c3=1)}} \label{subsec:sssy}

These terms can easily be computed using the results in section \ref{subsec:suxyv}.

\paragraph{Term $\expec\{z_{u}(i)z_{a}(i)z_{b}(i)y_{n}(i)\}$} $u = 1,\ldots,N|\cD_n|$, , $a,b = 1,\ldots,N|\cD_n|$

We adapt relations \eqref{eq:zzzy1}--\eqref{eq:zzzy2}, thus obtaining:
\begin{align}
    \expec\{z_{u}(i)z_{a}(i)z_{b}(i)y_{n}(i)\} 
    \nonumber 
    {}={} & 
    \expec\bigg\{ 
        \dfrac{(y_{m_1}(i) - y_{m_1}(\omega_p))}{\sigma^2} 
        \dfrac{(y_{m_2}(i) - y_{m_2}(\omega_q))}{\sigma^2} 
        \dfrac{(y_{m_3}(i) - y_{m_3}(\omega_r))}{\sigma^2} 
        y_n(i)
        \\ &
        \times \exp \left( -\dfrac{1}{2\sigma^2} \left(\Vert \by(i) - \by(\omega_p) \Vert^2 + \Vert \by(i) - \by(\omega_q) \Vert^2 + \Vert \by(i) - \by(\omega_r) \Vert^2\right)\right)
        \bigg\}
        \label{eq:zzzy3}\\ {}={} &
        \dfrac{1}{\sigma^6} \nu(\{h_i\}_{i=1}^8),\, \text{with }
    \begin{cases} 
    h_1=h_2=m_1=\left\lceil \frac{u}{|\cD_n|} \right\rceil \\ h_3=h_4-N=m_2=\left\lceil \frac{a}{|\cD_n|} \right\rceil \\ h_5=h_6-2N=m_3=\left\lceil \frac{b}{|\cD_n|} \right\rceil \\ h_7=h_8-3N=m_4=n
    \end{cases}\,,
    \begin{cases} 
    \omega_p=\modop(u-1,|\cD_n|)+1 \\ \omega_q=\modop(a-1,|\cD_n|)+1 \\ \omega_r=\modop(b-1,|\cD_n|)+1 \\ \omega_s=\bullet
    \end{cases}\,.
    \label{eq:zzzy4}
\end{align}

\paragraph{Term $\expec\{z_{u}(i)k_{a}(i)z_{b}(i)y_{n}(i)\}$} $u = 1,\ldots,N|\cD_n|$, $a = 1,\ldots,|\cD_n|, b = 1,\ldots,N|\cD_n|$

We adapt relations \eqref{eq:zkzy1}--\eqref{eq:zkzy2}, thus obtaining:
\begin{align}
    \expec\{z_{u}(i)k_{a}(i)z_{b}(i)y_{n}(i)\} 
    \nonumber
    {}={} &
    \expec\bigg\{ 
        \dfrac{(y_{m_1}(i) - y_{m_1}(\omega_p))}{\sigma^2} 
        \dfrac{(y_{m_3}(i) - y_{m_3}(\omega_r))}{\sigma^2} 
        y_n(i)
        \\ &
        \times \exp \left( -\dfrac{1}{2\sigma^2} \left(\Vert \by(i) - \by(\omega_p) \Vert^2 + \Vert \by(i) - \by(\omega_q) \Vert^2 + \Vert \by(i) - \by(\omega_r) \Vert^2\right)\right)
        \bigg\}
        \label{eq:zkzy3}\\ {}={} &
        \dfrac{1}{\sigma^4} \nu(\{h_i\}_{i=1}^8),\, \text{with }
    \begin{cases} 
    h_1=h_2=m_1=\left\lceil \frac{u}{|\cD_n|} \right\rceil \\ h_3=h_4-N=m_2=\bullet \\ h_5=h_6-2N=m_3=\left\lceil \frac{b}{|\cD_n|} \right\rceil \\ h_7=h_8-3N=m_4=n
    \end{cases}\,,
    \begin{cases} 
    \omega_p=\modop(u-1,|\cD_n|)+1 \\ \omega_q=a \\ \omega_r=\modop(b-1,|\cD_n|)+1 \\ \omega_s=\bullet
    \end{cases}\,.
    \label{eq:zkzy4}
\end{align}

\paragraph{Term $\expec\{z_{u}(i)k_{a}(i)k_{b}(i)y_{n}(i)\}$} $u = 1,\ldots,N|\cD_n|$, $a,b = 1,\ldots,|\cD_n|$

We adapt relations \eqref{eq:zkky1}--\eqref{eq:zkky2}, thus obtaining:
\begin{align}
    \expec\{z_{u}(i)k_{a}(i)k_{b}(i)y_{n}(i)\}
    \nonumber 
    {}={} &
    \expec\bigg\{ 
        \dfrac{(y_{m_1}(i) - y_{m_1}(\omega_p))}{\sigma^2} 
        y_n(i)
        \\ &
        \times\exp \left( -\dfrac{1}{2\sigma^2} \left(\Vert \by(i) - \by(\omega_p) \Vert^2 + \Vert \by(i) - \by(\omega_q) \Vert^2 + \Vert \by(i) - \by(\omega_r) \Vert^2\right)\right)
        \bigg\}
        \label{eq:zkky3}\\  {}={} &
        \dfrac{1}{\sigma^2} \nu(\{h_i\}_{i=1}^8),\, \text{with }
    \begin{cases} 
    h_1=h_2=m_1=\left\lceil \frac{u}{|\cD_n|} \right\rceil \\ h_3=h_4-N=m_2=\bullet \\ h_5=h_6-2N=m_3=\bullet \\ h_7=h_8-3N=m_4=n
    \end{cases}\,,
    \begin{cases} 
    \omega_p=\modop(u-1,|\cD_n|)+1 \\ \omega_q=a \\ \omega_r=b \\ \omega_s=\bullet
    \end{cases}\,.
    \label{eq:zkky4}
\end{align}

\paragraph{Term $\expec\{k_{u}(i)k_{a}(i)k_{b}(i)y_{n}(i)\}$} $u= 1,\ldots, |\cD_n|$, $a,b = 1,\ldots,|\cD_n|$

We adapt relations \eqref{eq:kkky1}--\eqref{eq:kkky2}, thus obtaining:
\begin{align}
    \expec\{k_{u}(i)k_{a}(i)k_{b}(i)k_{v}(i)\}
    {}={} & y_n(i) \exp \left( -\dfrac{1}{2\sigma^2} \left(\Vert \by(i) - \by(\omega_p) \Vert^2 + \Vert \by(i) - \by(\omega_q) \Vert^2 + \Vert \by(i) - \by(\omega_r) \Vert^2 \right)\right) \bigg\} 
    \\ {}={} & \nu(\{h_i\}_{i=1}^8),\, \text{with }
    \begin{cases} 
    h_1=h_2=m_1=\bullet \\ h_3=h_4-N=m_2=\bullet \\ h_5=h_6-2N=m_3=\bullet \\ h_7=h_8-3N=m_4=n
    \end{cases}\,,
    \begin{cases} 
    \omega_p=u \\ \omega_q=a \\ \omega_r=b \\ \omega_s=\bullet
    \end{cases}\,.
\end{align}


\section{Cases corresponding to \texorpdfstring{$\expec\{s_a(i) s_{b}(i) y_n^2(i)\}$ ($c_1=c_2=0,c_3=1$)}{E\{sa(i)sb(i)yn(i)yn(i)\} (c1=c2=0,c3=1)}} \label{subsec:ssyy}

These terms can easily be computed using the results in section \ref{subsec:sssy}.

\paragraph{Term $\expec\{y_{n}(i)z_{a}(i)z_{b}(i)y_{n}(i)\}$}$\,$, $a,b = 1,\ldots,N|\cD_n|$

We adapt relations \eqref{eq:zzzy3}--\eqref{eq:zzzy4}, thus obtaining:
\begin{align}
    \expec\{y_{n}(i)z_{a}(i)z_{b}(i)y_{n}(i)\}
    \nonumber
    {}={} & 
    \expec\bigg\{ 
        y_n(i)
        \dfrac{(y_{m_2}(i) - y_{m_2}(\omega_q))}{\sigma^2} 
        \dfrac{(y_{m_3}(i) - y_{m_3}(\omega_r))}{\sigma^2} 
        y_n(i)
        \\ &
        \times \exp \left( -\dfrac{1}{2\sigma^2} \left(\Vert \by(i) - \by(\omega_q) \Vert^2 + \Vert \by(i) - \by(\omega_r) \Vert^2\right)\right)
        \bigg\}
        \label{eq:zzzy5}\\  {}={} &
        \dfrac{1}{\sigma^4} \nu(\{h_i\}_{i=1}^8),\, \text{with }
    \begin{cases} 
    h_1=h_2=m_1=n \\ h_3=h_4-N=m_2=\left\lceil \frac{a}{|\cD_n|} \right\rceil \\ h_5=h_6-2N=m_3=\left\lceil \frac{b}{|\cD_n|} \right\rceil \\ h_7=h_8-3N=m_4=n
    \end{cases}\,,
    \begin{cases} 
    \omega_p=\bullet \\ \omega_q=\modop(a-1,|\cD_n|)+1 \\ \omega_r=\modop(b-1,|\cD_n|)+1 \\ \omega_s=\bullet
    \end{cases}\,.
    \label{eq:zzzy6}
\end{align}

\paragraph{Term $\expec\{y_{n}(i)k_{a}(i)z_{b}(i)y_{n}(i)\}$}$\,$ $a = 1,\ldots,|\cD_n|$, $b = 1,\ldots,N|\cD_n|$

We adapt relations \eqref{eq:zkzy3}--\eqref{eq:zkzy4}, thus obtaining:
\begin{align}
    \expec\{y_{n}(i)k_{a}(i)z_{b}(i)y_{n}(i)\} 
    & = 
    \expec\bigg\{ 
        y_n(i) 
        \dfrac{(y_{m_3}(i) - y_{m_3}(\omega_r))}{\sigma^2} 
        y_n(i)
        \exp \left( -\dfrac{1}{2\sigma^2} \left(\Vert \by(i) - \by(\omega_q) \Vert^2 + \Vert \by(i) - \by(\omega_r) \Vert^2\right)\right)
        \bigg\}
        \\ & =
        \dfrac{1}{\sigma^2} \nu(\{h_i\}_{i=1}^8),\, \text{with }
    \begin{cases} 
    h_1=h_2=m_1=n \\ h_3=h_4-N=m_2=\bullet \\ h_5=h_6-2N=m_3=\left\lceil \frac{b}{|\cD_n|} \right\rceil \\ h_7=h_8-3N=m_4=n
    \end{cases}\,,
    \begin{cases} 
    \omega_p=\bullet \\ \omega_q=a \\ \omega_r=\modop(b-1,|\cD_n|)+1 \\ \omega_s=\bullet
    \end{cases}\,.
\end{align}

\paragraph{Term $\expec\{y_{n}(i)k_{a}(i)k_{b}(i)y_{n}(i)\}$}$\,$ $a,b = 1,\ldots,|\cD_n|$

We adapt relations \eqref{eq:zkky3}--\eqref{eq:zkky4}, thus obtaining:
\begin{align}
    \expec\{y_{n}(i)k_{a}(i)k_{a}(i)y_{n}(i)\} & = 
    \expec\bigg\{ 
        y_n(i)
        y_n(i)
        \exp \left( -\dfrac{1}{2\sigma^2} \left(\Vert \by(i) - \by(\omega_q) \Vert^2 + \Vert \by(i) - \by(\omega_r) \Vert^2\right)\right)
        \bigg\}
        \label{eq:zkzy5}\\ & =
        \nu(\{h_i\}_{i=1}^8),\, \text{with }
    \begin{cases} 
    h_1=h_2=m_1=n \\ h_3=h_4-N=m_2=\bullet \\ h_5=h_6-2N=m_3=\bullet \\ h_7=h_8-3N=m_4=n
    \end{cases}\,,
    \begin{cases} 
    \omega_p=\bullet \\ \omega_q=a \\ \omega_r=b \\ \omega_s=\bullet
    \end{cases}\,.
    \label{eq:zkzy6}
\end{align}


\section{Cases corresponding to \texorpdfstring{$\expec\{s_b(i) y_n(i)\}$ ($c_1=c_2=c_3=0$)}{E\{sb(i)yn(i)\} (c1=c2=0=c3=0)}} \label{subsec:sy}

These terms can easily be computed using the results in section \ref{subsec:ssyy}.

\paragraph{Term $\expec\{z_{b}(i)y_{n}(i)\}$}$\,$, $b = 1,\ldots,N|\cD_n|$

We adapt relations \eqref{eq:zzzy5}--\eqref{eq:zzzy6}, thus obtaining:
\begin{align}
    \expec\{z_{b}(i)y_{n}(i)\} &= 
    \expec\bigg\{ 
        \dfrac{(y_{m_3}(i) - y_{m_3}(\omega_r))}{\sigma^2} 
        y_n(i)
        \exp \left( -\dfrac{1}{2\sigma^2} \left(\Vert \by(i) - \by(\omega_r) \Vert^2\right)\right)
        \bigg\}
        \\ & =
        \dfrac{1}{\sigma^2} \nu(\{h_i\}_{i=1}^8),\, \text{with }
    \begin{cases} 
    h_1=h_2=m_1=\bullet \\ h_3=h_4-N=m_2=\bullet \\ h_5=h_6-2N=m_3=\left\lceil \frac{y}{|\cD_n|} \right\rceil \\ h_7=h_8-3N=m_4=n
    \end{cases}\,,
    \begin{cases} 
    \omega_p=\bullet \\ \omega_q=\bullet \\ \omega_r=\modop(b-1,|\cD_n|)+1 \\ \omega_s=\bullet
    \end{cases}\,.
\end{align}

\paragraph{Term $\expec\{k_{y}(i)y_{n}(i)\}$}$\,$, $b = 1,\ldots,|\cD_n|$

We adapt relations \eqref{eq:zkzy5}--\eqref{eq:zkzy6}, thus obtaining:
\begin{align}
    \expec\{k_{b}(i)y_{n}(i)\} & = 
    \expec\bigg\{ 
        y_n(i)
        \exp \left( -\dfrac{1}{2\sigma^2} \left(\Vert \by(i) - \by(\omega_r) \Vert^2\right)\right)
        \bigg\} \\& =
        \nu(\{h_i\}_{i=1}^8),\, \text{with }
    \begin{cases} 
    h_1=h_2=m_1=\bullet \\ h_3=h_4-N=m_2=\bullet \\ h_5=h_6-2N=m_3=\bullet \\ h_7=h_8-3N=m_4=n
    \end{cases}\,,
    \begin{cases} 
    \omega_p=\bullet \\ \omega_q=\bullet \\ \omega_r=b \\ \omega_s=\bullet
    \end{cases}\,.
\end{align}

\fi

\end{document}